\documentclass{article}

\usepackage[nonatbib, final]{neurips_2021}
\usepackage{wrapfig}

\usepackage[utf8]{inputenc} 
\usepackage[T1]{fontenc}    
\usepackage{hyperref}       
\usepackage{url}            
\usepackage{booktabs}       
\usepackage{amsfonts}       
\usepackage{nicefrac}       
\usepackage{microtype}      
\usepackage{xcolor}         

\usepackage[numbers]{natbib}

\usepackage{algorithm}

\usepackage{bm}
\usepackage{amsfonts}
\usepackage{amsmath}
\usepackage{amssymb}
\usepackage[noend]{algpseudocode}
\usepackage{bm}
\usepackage{nicefrac}
\usepackage{microtype}
\usepackage{graphicx}
\usepackage{mathabx}
\usepackage{subcaption}
\usepackage{booktabs} 
\usepackage{amsthm}

\usepackage{hyperref}
\newtheorem{proposition}{Proposition}
\newtheorem{theorem}{Theorem}

\algdef{SE}[DOWHILE]{Do}{doWhile}{\algorithmicdo}[1]{\algorithmicwhile\ #1}

\usepackage[
  separate-uncertainty = true,
  multi-part-units = repeat
]{siunitx}

\title{XDO: A Double Oracle Algorithm for Extensive-Form Games}

\author{%
  Stephen McAleer\\
  Department of Computer Science\\
  University of California, Irvine\\
  \texttt{smcaleer@uci.edu} \\
   \And
  John Lanier\\
  Department of Computer Science\\
  University of California, Irvine\\
  \texttt{jblanier@uci.edu}
  \And
  Kevin A. Wang\\
  Department of Computer Science\\
  University of California, Irvine\\
  \texttt{kevinwang@kevinwang.us}\\
  \And
  Pierre Baldi\\
  Department of Computer Science\\
  University of California, Irvine\\
  \texttt{pfbaldi@ics.uci.edu}\\
  \And
  Roy Fox\\
  Department of Computer Science\\
  University of California, Irvine\\
  \texttt{royf@uci.edu}
}

\begin{document}

\maketitle
\begin{abstract}
Policy Space Response Oracles (PSRO) is a reinforcement learning (RL) algorithm for two-player zero-sum games that has been empirically shown to find approximate Nash equilibria in large games. Although PSRO is guaranteed to converge to an approximate Nash equilibrium and can handle continuous actions, it may take an exponential number of iterations as the number of information states (infostates) grows. We propose Extensive-Form Double Oracle (XDO), an extensive-form double oracle algorithm for two-player zero-sum games that is guaranteed to converge to an approximate Nash equilibrium \textit{linearly} in the number of infostates. Unlike PSRO, which mixes best responses at the root of the game, XDO mixes best responses at every infostate. We also introduce Neural XDO (NXDO), where the best response is learned through deep RL. In tabular experiments on Leduc poker, we find that XDO achieves an approximate Nash equilibrium in a number of iterations an order of magnitude smaller than PSRO. Experiments on a modified Leduc poker game and Oshi-Zumo show that tabular XDO achieves a lower exploitability than CFR with the same amount of computation. We also find that NXDO outperforms PSRO and NFSP on a sequential multidimensional continuous-action game. NXDO is the first deep RL method that can find an approximate Nash equilibrium in high-dimensional continuous-action sequential games. Experiment code is available at \url{https://github.com/indylab/nxdo}.
\end{abstract}

\section{Introduction}
Policy Space Response Oracles (PSRO)~\citep{psro} is a reinforcement learning (RL) method for finding approximate Nash equilibria (NE) in large two-player zero-sum games. Methods based on PSRO have recently achieved state-of-the-art performance on large imperfect-information two-player zero-sum games such as Starcraft~\citep{alphastar} and Stratego~\citep{mcaleer2020pipeline}. One major benifit of PSRO versus other deep RL methods for two-player zero-sum games is that it is naturally compatible with games that have continuous actions. The only other deep RL method compatible with continuous actions, self play, is not guaranteed to converge to a Nash equilibrium even in small games like Rock Paper Scissors.
Despite the empirical success of PSRO, 
in the worst case, PSRO may need to expand all pure strategies in the normal form of the game, which grows exponentially in the number of information states (infostates).
The reason for this is that PSRO is based on the Double Oracle algorithm for normal-form games~\citep{double_oracle}, and a mixture of normal-form pure strategies is an inefficient representation of extensive-form policies. 

In this work, we propose a new double oracle algorithm, Extensive-Form Double Oracle (XDO), that is designed for extensive-form (sequential) games. 
Like PSRO, XDO keeps a population of pure strategies. At every iteration, XDO creates a restricted game by only considering actions that are chosen by at least one strategy in the population. This restricted game is then approximately solved via an extensive-form game solver, such as Counterfactual Regret Minimization (CFR)~\citep{cfr} or Fictitious Play (FP)~\citep{fp}, to find a meta-NE, which is extended to the full game by taking arbitrary actions at infostates not encountered in the restricted game. Next, a best response (BR) to the restricted game meta-NE is computed and added to the population. XDO can be viewed as a version of PSRO where the restricted game allows mixing population strategies not only at the root of the game, but at every infostate.

XDO is guaranteed to converge to an approximate NE in a number of iterations that is linear in the number of infostates, while PSRO may require a number of iterations exponential in the number of infostates.
Furthermore, on a worst-case family of games for the lower bound on the number of PSRO iterations, we show that XDO converges in a number of iterations that does not grow with the number of infostates, and grows only linearly with the number of actions at each infostate.

We also introduce a neural version of XDO, called Neural XDO (NXDO). NXDO can be used in games that are large enough to benefit from the generalization over infostates induced by neural-network strategies. 
NXDO learns approximate BRs through any deep reinforcement learning algorithm. The restricted game consists of meta-actions, each selecting a population policy to play the next action. 
This restricted game is then solved through any neural extensive-form game solver, such as NFSP~\citep{nfsp} or Deep CFR~\citep{deep_cfr}. In our experiments, we use PPO \citep{schulman2017proximal} or DDQN \citep{van2016deep} for the approximate BR and NFSP as the restricted game solver. Although convergence guarantees may not apply in such cases, like PSRO, NXDO is compatible with continuous action spaces.


In games with a large number of actions, NXDO and PSRO effectively prune the game tree and outperform methods such as Deep CFR and NFSP, which cannot be applied at all with continuous actions. Additionally, because PSRO might require an exponential number of pure strategies, NXDO outperforms PSRO on games that require mixing over multiple timesteps. 
To demonstrate the effectiveness of our approach on these types of games, we run experiments on two sets of environments. The first, $m$-Clone Leduc, is similar to Leduc poker but with every call, fold, and bet action duplicated $m$ times. 
The second, the Loss Game, is a sequential continuous-action multidimensional optimization game in which agents simultaneously adjust parameters to maximize or minimize a complex loss function. 
We show that tabular XDO greatly outperforms PSRO, CFR, and XFP~\citep{xfp} on $m$-Clone Leduc. We also show that NXDO outperforms both PSRO and NFSP on $m$-Clone Leduc and on the continuous-action Loss Game, where NFSP is provided a binned discrete action space. 

To summarize, our contributions are as follows:
\begin{itemize}
    \item We present a tabular extensive-form double oracle algorithm, XDO, that terminates in a linear number of iterations in the number of infostates.
    \item We present a neural version of XDO, NXDO, that outperforms PSRO and NFSP on both modified Leduc poker and sequential continuous-action games. NXDO is the first method that can find an approximate NE in high-dimensional continuous-action sequential games. 
\end{itemize}

\section{Background}
\subsection{Extensive-Form Games}
We consider partially-observable stochastic games \citep{hansen2004dynamic} which correspond to perfect-recall extensive-form games (from here on referred to as extensive-form games). An extensive-form game progresses through a sequence of player actions, and has a \textbf{world state} $w \in \mathcal{W}$ at each step. 
In an $N$-player game, $\mathcal{A} = \mathcal{A}_1 \times \cdots \times \mathcal{A}_N$ is the space of joint actions for the players. $\mathcal{A}_i(w)$ denotes the set of legal actions for player $i \in \mathcal{N} = \{1, \ldots, N\}$ at world state $w$ and $a = (a_1, \ldots, a_N) \in \mathcal{A}$ denotes a joint action. At each world state, after the players choose a joint action, a transition function $\mathcal{T}(w, a) \in \Delta^\mathcal{W}$ determines the probability distribution of the next world state $w'$. Upon transition from world state $w$ to $w'$ via joint action $a$, player $i$ makes an \textbf{observation} $o_i = \mathcal{O}_i(w,a,w')$. In each world state $w$, player $i$ receives a reward $\mathcal{R}_i(w)$.

A \textbf{history} is a sequence of actions and world states, denoted $h = (w^0, a^0, w^1, a^1, \ldots, w^t)$, where $w^0$ is the known initial world state of the game. $\mathcal{R}_i(h)$ and $\mathcal{A}_i(h)$ are, respectively, the reward and set of legal actions for player $i$ in the last world state of a history $h$. An \textbf{infostate} for player $i$, denoted by $s_i$, is a sequence of that player's observations and actions up until that time $s_i(h) = (a_i^0, o_i^1, a_i^1, \ldots, o_i^t)$. Define the set of all infostates for player $i$ to be $\mathcal{I}_i$. 
The set of histories that correspond to an infostate $s_i$ is denoted $\mathcal{H}(s_i) = \{ h: s_i(h) = s_i \}$, and it is assumed that they all share the same set of legal actions $\mathcal{A}_i(s_i(h)) = \mathcal{A}_i(h)$. 

A player's \textbf{policy} $\pi_i$ 
is a function mapping from an infostate to a probability distribution over actions. A \textbf{policy profile} $\pi$ is a tuple $(\pi_1, \ldots, \pi_N)$. All players other than $i$ are denoted $-i$, and their policies are jointly denoted $\pi_{-i}$. A policy for a history $h$ is denoted $\pi_i(h) = \pi_i(s_i(h))$ and $\pi(h)$ is the corresponding policy profile. 
We also define the transition function $\mathcal{T}(h, a_i, \pi_{-i}) \in \Delta^{\mathcal{W}}$ as a function drawing actions for $-i$ from $\pi_{-i}$ to form $a = (a_i, a_{-i})$ and to then sample the next world state $w'$ from $\mathcal{T}(w, a)$, where $w$ is the last world state in $h$.

The \textbf{expected value (EV)} $v_i^{\pi}(h)$ for player $i$ is the expected sum of future rewards for player $i$ in history $h$, when all players play policy profile $\pi$. The EV for an infostate $s_i$ is denoted $v_i^{\pi}(s_i)$ and the EV for the entire game is denoted $v_i(\pi)$. A \textbf{two-player zero-sum} game has $v_1(\pi) + v_2(\pi) = 0$ for all policy profiles $\pi$. The EV for an action in an infostate is denoted $v_i^{\pi}(s_i,a_i)$. A \textbf{Nash equilibrium (NE)} is a policy profile such that, if all players played their NE policy, no player could achieve higher EV by deviating from it. Formally, $\pi^*$ is a NE if $v_i(\pi^*) = \max_{\pi_i}v_i(\pi_i, \pi^*_{-i})$ for each player $i$.

The \textbf{exploitability} $e(\pi)$ of a policy profile $\pi$ is defined as $e(\pi) = \sum_{i \in \mathcal{N}} \max_{\pi'_i}v_i(\pi'_i, \pi_{-i})$. A \textbf{best response (BR)} policy $\mathbb{BR}_i(\pi_{-i})$ for player $i$ to a policy $\pi_{-i}$ is a policy that maximally exploits $\pi_{-i}$: $\mathbb{BR}_i(\pi_{-i}) = \arg\max_{\pi_i}v_i(\pi_i, \pi_{-i})$. An \textbf{$\bm{\epsilon}$-best response ($\bm{\epsilon}$-BR)} policy $\mathbb{BR}^\epsilon_i(\pi_{-i})$ for player $i$ to a policy $\pi_{-i}$ is a policy that is at most $\epsilon$ worse for player $i$ than the best response: $v_i(\mathbb{BR}^\epsilon_i(\pi_{-i}), \pi_{-i}) \ge v_i(\mathbb{BR}_i(\pi_{-i}), \pi_{-i}) - \epsilon$. An \textbf{$\bm{\epsilon}$-Nash equilibrium ($\bm{\epsilon}$-NE)} is a policy profile $\pi$ in which, for each player $i$, $\pi_i$ is an $\epsilon$-BR to $\pi_{-i}$. 


A \textbf{normal-form game} is a single-step extensive-form game. An extensive-form game induces a normal-form game in which the legal actions for player $i$ are its deterministic policies $\bigtimes_{s_i \in \mathcal{I}_i} \mathcal{A}_i(s_i)$. These deterministic policies are called \textbf{pure strategies} of the normal-form game. Since each deterministic policy specifies one action at every infostate, there are an exponential number of pure strategies in the number of infostates. A \textbf{mixed strategy} is a distribution over a player's pure strategies. Two policies $\pi_i^1$ and $\pi_i^2$ for player $i$ are said to be \textbf{realization-equivalent} if for any fixed strategy profile of the other player, both $\pi_i^1$ and $\pi_i^2$, define the same probability distribution over the states of the game. 

\begin{theorem}[Kuhn's Theorem \citep{kuhn1953contributions}]
Any mixed strategy in the normal form of a game is realization equivalent to a policy in the extensive form of that game, and vice versa.
\end{theorem}





\section{Related Work}
There has been much recent work on non-game-theoretic multi-agent RL \citep{foerster2018counterfactual, lowe2017multi, rashid2018qmix, bansal2017emergent}. Most of this work focuses on games with more than two players such as multi-agent cooperative games or mixed competitive-cooperative scenarios. In cooperative environments, self-play has empirically been shown to find an approximate NE \citep{lowe2017multi, majumdar2020evolutionary}, but can be brittle when cooperating with agents it hasn't trained with \citep{psro}. Self-play reinforcement learning has achieved expert level performance on video games \citep{alphastar, dota, pbt}, but is not guaranteed to converge to an approximate NE.

Extensive-form fictitious play (XFP) \citep{xfp} and counterfactual regret minimization (CFR) \citep{cfr} extend Fictitious Play (FP)~\citep{fp} and regret matching~\citep{rm}, respectively, to extensive-form games. 
Deep CFR~\citep{deep_cfr, steinberger2019single, li2018double} is a general method that trains a neural network on a buffer of counterfactual values. However, Deep CFR uses external sampling, which may be impractical for games with a large branching factor, such as Stratego and Barrage Stratego. DREAM \citep{steinberger2020dream} and ARMAC \citep{gruslys2020advantage} are model-free regret-based deep learning approaches. DREAM and ARMAC have achieved good results in poker games, but since they are based on MCCFR, like Deep CFR, they will not scale to games with continuous actions.

Our work is related to pruning approaches~\citep{brown2015regret, brown2017dynamic}. These methods start with all actions and sequentially remove actions that have low expected value. XDO instead starts with no actions and sequentially adds actions. Our work is also related to methods that automatically find abstractions \citep{brown2015simultaneous, vcermak2017algorithm}. 

Close to our work, \citet{bosansky2014exact} develop a sequence-form double oracle (SDO) algorithm.
The SDO algorithm iteratively adds sequence-form BRs to a population and then computes a meta-Nash on a restricted sequence-form game where only sequences in the population are allowed. In contrast, XDO iteratively adds extensive-form BRs to a population and then computes a meta-Nash on a restricted extensive form game where only actions in the population are allowed. DO, SDO, and XDO are fundamentally different because they operate on the normal form, sequence form, and extensive form, respectively. We give a detailed description of the difference between XDO and SDO in the supplementary materials.

\subsection{Neural Fictitious Self Play (NFSP)}
Neural Fictitious Self Play (NFSP) \citep{nfsp} approximates XFP by progressively training a best response against an average of all past policies using reinforcement learning. The average policy is represented by a neural network and is trained via supervised learning using a replay buffer of past best response actions. Each episode, both players either play from their best response policy with probability $\eta = 0.1$ or with their average policy with probability $1-\eta$. This experience is then added to the best response circular replay buffer and is used to train the best response for both players with off-policy DQN. If a player plays with their best response policy, the data is also added to the average policy reservoir replay buffer and is used to train the average policy via supervised learning. 

\subsection{Policy Space Response Oracles (PSRO)}
The Double Oracle algorithm \citep{double_oracle} is an algorithm for finding a NE in normal-form games. The algorithm works by keeping a population of policies $\Pi^t$ at time $t$. Each iteration a meta-Nash Equilibrium (meta-NE) $\pi^{*,t}$ is computed for the game restricted to policies in $\Pi^t$. Then, a best response to this meta-NE for each player $\mathbb{BR}_i(\pi^{*,t}_{-i})$ is computed and added to the population $\Pi_i^{t+1} = \Pi_i^t \cup \{\mathbb{BR}_i(\pi^{*,t}_{-i}) \}$ for $i \in \{1, 2\}$.

Policy Space Response Oracles (PSRO) \citep{psro} approximates the Double Oracle algorithm. The meta-NE is computed on the empirical game matrix $U^\Pi$, given by having each policy in the population $\Pi$ play each other policy and tracking average utility in a payoff matrix. In each iteration, an approximate best response to the current meta-NE over the policies is computed via any reinforcement learning algorithm. Pipeline PSRO parallelizes PSRO with convergence guarantees \citep{mcaleer2020pipeline}.

\subsubsection{PSRO Hard Instance}
A primary issue with PSRO is that it is based on a normal-form algorithm, and the number of pure strategies in a normal-form representation of an extensive-form game is exponential in the number of infostates. In contrast, our approach implements the double oracle algorithm directly in the extensive-form game, overcoming this problem and terminating in a linear number of iterations in the number of infostates. The following example helps illustrate this point.

Consider the game in Figure \ref{fig:game}. In this game, first, player 1 chooses which Rock Paper Scissors (RPS) game both players play. After player 1 chooses the RPS game, both players know which RPS game they are playing. Then both players simultaneously play an action in that RPS game. There are 6 pure strategies for player 1, denoted R1, P1, S1, R2, P2, S2. But there are 9 pure strategies for player 2. A pure strategy for player 2 specifies what move they play at each infostate. If player 2 played Rock in the first infostate and Paper in the second, that pure strategy is denoted RP. Note that if we generalize this game by including more RPS games and more actions in each game, the number of pure strategies for player 2 will be $|A|^{|\mathcal{I}|}$, where $|A|$ is the number of actions and $|\mathcal{I}|$ is the number of RPS games.


We conduct an experiment where we run PSRO with oracle BRs on this game with a random starting population each time. We find that PSRO expands all 9 row (player 2) pure strategies the majority of the time, expanding all 9 strategies in 122 out of 150 trials. These results are shown in Figure \ref{fig:hist}. We also find that the column player (player 1) expands all 6 pure strategies in all 150 trials. 

\begin{figure} 
\centering
\begin{subfigure}{.40\textwidth}
  \centering
  \includegraphics[width=\linewidth]{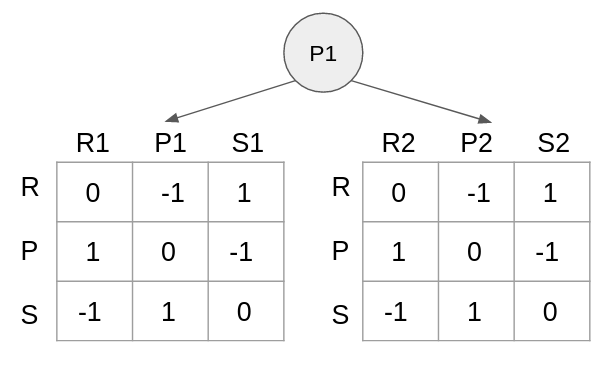}
  \caption{}
  \label{fig:ef}
\end{subfigure}%
\begin{subfigure}{.22\textwidth}
  \centering
  \includegraphics[width=.9\linewidth]{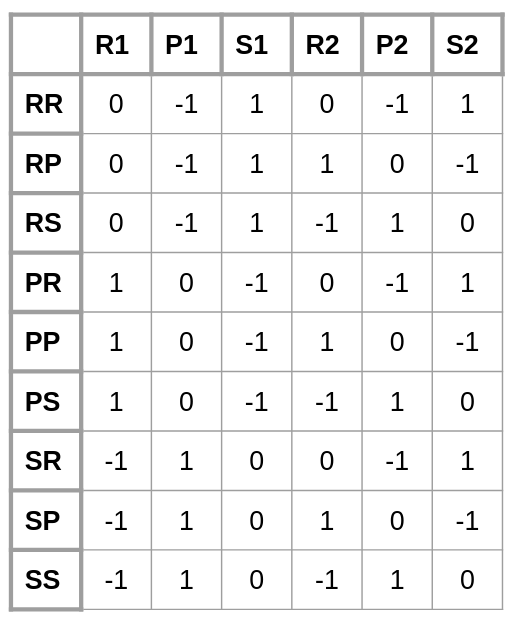}
  \caption{}
  \label{fig:nf}
\end{subfigure}
\begin{subfigure}{.35\textwidth}
    \includegraphics[width=\linewidth]{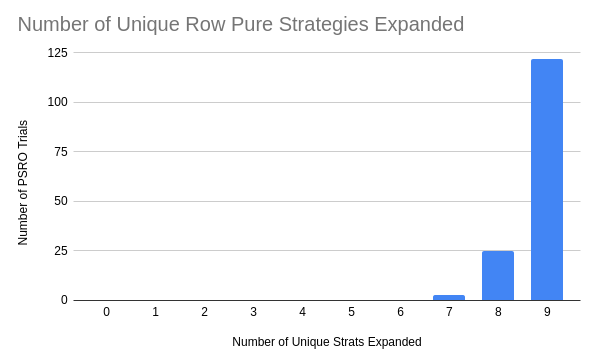}
    \caption{}
    \label{fig:hist}
\end{subfigure}

\caption{PSRO hard instance. (a) Player 1 first chooses which RPS game both players play. Both players know which RPS game they are playing. Then both players simultaneously make their move. (b) The normal form game. Player 2 has 9 pure strategies. (c) The proportion of PSRO trials that expanded each possible number of pure strategies for player 2. In the majority of trials, PSRO had to expand all possible pure strategies.}

\label{fig:game}
\end{figure}

\section{Extensive-Form Double Oracle (XDO)}

 

We propose Extensive-Form Double Oracle (XDO), a double-oracle (DO) algorithm designed for two-player zero-sum extensive-form games (Algorithm \ref{alg:XDO}). As in other DO algorithms, XDO maintains a population of pure strategies, and in each iteration computes a meta-NE of this population. Then the algorithm finds a best response (BR) to the meta-NE and adds it to the population.

In XDO, the population induces a different restricted game, and therefore a different population meta-NE, than in PSRO~\citep{psro}. In PSRO, a restricted normal-form game is induced by the empirical payoff matrix of population strategies. In XDO, a restricted extensive-form game is induced through a transformation on the original base extensive-form game that restricts the allowed actions at each infostate to only those suggested by any strategy in the population.


XDO uses a tabular method such as CFR~\citep{cfr} or XFP~\citep{xfp} to solve the restricted game. The algorithm terminates after an iteration in which neither of the players finds a BR that outperforms the meta-NE. When this happens, the meta-NE policies are approximate BRs to each other in the original game as well, and the meta-NE is therefore an approximate NE of the original game.

Importantly, at each but the final iteration of XDO, at least one player adds some new action at some non-terminal infostate, because a BR cannot outperform the meta-NE with only restricted-game actions.
The number of iterations that XDO takes to terminate is therefore at most the number of infostates, including terminal ones. In contrast, the best known guarantee for the number of iterations that PSRO takes to terminate (Proposition \ref{prop:psro}) is exponential in the number of infostates, because PSRO may need to add all pure strategies to the population. Moreover, computing the meta-NE in PSRO may become intractable in later iterations as the population size increases, while in XDO it is bounded by the unrestricted game.



Formally, XDO keeps a population of pure strategies $\Pi^t$ at time $t$. Each iteration, a restricted extensive-form game is created and a NE to the restricted game is computed. The restricted game is created by taking the original game and restricting the actions at every infostate $s_i$ to be only the actions where there exists a policy in the population $\Pi^t$ that chooses that action at that infostate: 
\begin{equation}\label{restricted_game}
\mathcal{A}^r_i(s_i) = \{a \in \mathcal{A}_i(s_i): \exists \pi_i \in \Pi_i^t\ \mathrm{s.t.}\ \pi_i(s_i, a) = 1\}.
\end{equation}

An $\epsilon$-NE policy $\pi^{r*}$ is then computed in this restricted game via a tabular method such as CFR and is extended to the full game by defining arbitrary actions on infostates not encountered in the restricted game. Next, BRs to this restricted game meta-NE $\mathbb{BR}_1(\pi^{r*}_2)$ and $\mathbb{BR}_2(\pi^{r*}_1)$ are computed via an oracle. These BRs are then added to the population of policies: $\Pi^{t+1}_i = \Pi^t_i \cup \mathbb{BR}_i(\pi^{r*}_{-i})$ for $i \in \{1, 2\}$.

The algorithm terminates when neither player benefits more than $\epsilon$ from deviating from the meta-NE to the BR, indicating that the meta-NE is an $\epsilon$-NE also in the original game (Proposition~\ref{prop:converge}). 

\begin{algorithm}[t]
    \caption{XDO}
    \label{alg:XDO} 
\begin{algorithmic}[1]

\algblockdefx{MRepeat}{EndRepeat}{\textbf{repeat}}{}
\algnotext{EndRepeat}

\State Input: initial population $\Pi^0$
\MRepeat
    \State Define restricted game for $\Pi^t$ via eq. (\ref{restricted_game})
    \State Get $\epsilon$-NE policy $\pi^{r*}$ of restricted game
    \State Find $\mathbb{BR}_i(\pi^{r*}_{-i})$ for $i \in \{1,2\}$ 
    \If{$v_i(\mathbb{BR}_i(\pi^{r*}_{-i}), \pi^{r*}_{-i}) \le v_i(\pi^{r*}) + \epsilon$ for both $i$}
        \State Terminate
    \EndIf
    \State $\Pi^{t+1}_i = \Pi^t_i \cup \mathbb{BR}_i(\pi^{r*}_{-i})$ for $i \in \{1, 2\}$
\EndRepeat
\end{algorithmic}
\end{algorithm}
To illustrate how XDO works, we demonstrate a simple game in Figure \ref{fig:xdo_game}. The algorithm starts with empty populations. At the first iteration (left diagram), player 1 adds a BR that plays Left at the first infostate (the root) and Right at the second one. Player 2 simultaneously adds a BR that plays Right at their single infostate. The restricted game now consists of only these added actions. At the second iteration (middle diagram), player 1 adds a BR that plays Right at both infostates, and player 2's BR still plays Right. The restricted game now includes both actions for the root infostate, but only Right is in the meta-NE. Next, in the third iteration (right diagram), player 1 keeps the same BR, while player 2's BR plays Left. In the meta-NE of this final restricted game, player 1 plays Left and Right with equal probability at the first infostate, and player 2 plays Left with probability 0.37 and Right with probability 0.63. Since the BRs to this meta-NE do not add any new actions, XDO terminates, and the meta-NE is the NE for the full game. Note that in this example, most actions are needed to find a NE. In games like this, it would be faster to simply solve the original game from the beginning. However, certain games, such as those in our experiments, have Nash equilibria that only need to mix over a small subset of actions \citep{schmid2014bounding}, in which case solving the XDO restricted game will be much faster than solving the original game.

\subsection{Theoretical Considerations}

In this section, we present a theoretical analysis of XDO and compare it with PSRO. Our first proposition states that, when XDO terminates, the final meta-NE of the restricted game is an approximate NE of the full game.

\begin{proposition}\label{prop:converge}
In XDO with an $\epsilon_1$-BR oracle, let $\pi^{r*}$ be the final $\epsilon_2$-NE in the restricted game. Then $\pi^{r*}$ is an $(\epsilon_1 + \epsilon_2)$-NE in the full game.
\end{proposition}
\begin{proof}
Due to limited space, all proofs are contained in the supplementary materials.
\end{proof}

\begin{figure*}[]
\centering
\includegraphics[width=1.0\columnwidth]{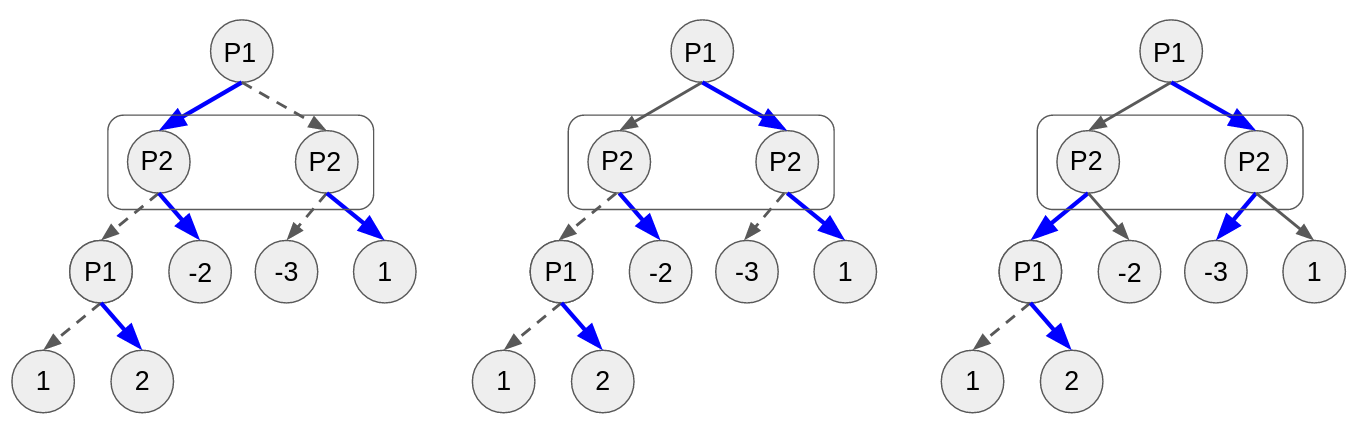}
\caption{Three iterations of XDO (left to right). In these extensive-form game diagrams, player 1 (P1) plays at the root, then P2 plays without knowing P1's action, and if both played Left P1 plays another action. P1's reward is number at the reached leaf. Actions in the restricted game are solid, vs. dashed outside the restricted game. Meta-NE actions are blue, vs. black not in the meta-NE.}
\label{fig:xdo_game}
\end{figure*}

The next two propositions show an exponential gap in the known guarantees for the number of iterations in which PSRO and XDO terminate.
If each non-terminal infostate allows $A$ different actions, PSRO is guaranteed to terminate in $\sum_i A^{|\bar{\mathcal{I}}_i|}$ iterations, where $\bar{\mathcal{I}}_i$ is the set of non-terminal infostates for player $i$, while XDO is guaranteed to terminate in $\sum_i |\mathcal{I}_i|$ iterations.

\begin{proposition}\label{prop:psro}
Normal-form DO terminates in at most $\sum_i \prod_{s_i \in \mathcal{I}_i} |\mathcal{A}_i(s_i)|$ iterations.
\end{proposition}

\begin{proposition}\label{prop:xdo}
XDO terminates in at most $\sum_i |\mathcal{I}_i|$ iterations. 
\end{proposition}

\paragraph{Tightness of the guarantees.} The guarantees in Proposition \ref{prop:psro} and Proposition \ref{prop:xdo} are tight in the sense that they are achieved in some games, but more nuanced analysis is required to identify easier cases where these bounds overestimate the complexity of the algorithms. Both PSRO and XDO often outperform these guarantees and terminate in fewer iterations. A case in which PSRO 
expands all pure normal-form strategies of an extensive-form game is described in the supplementary materials.

\begin{wrapfigure}{r}{0.3\textwidth}
    \centering
    \includegraphics[width=45mm]{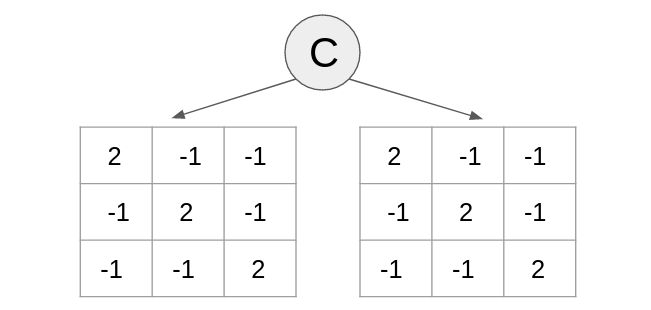}
    \caption{A $2$-GMP game with $n=3$ actions. The chance node selects uniformly at random which generalized matching pennies game is played. Both players know which stage game they play.}
    \label{fig:example_game}
\end{wrapfigure}

\paragraph{XDO can add multiple actions in each iteration.} In practice, XDO often outperforms the guarantee of Proposition \ref{prop:xdo} because it adds multiple actions in each iteration. Here we present and analyze a family of games in which XDO terminates in asymptotically fewer iterations than suggested by the bound in Proposition \ref{prop:xdo}.

In a generalized matching pennies (GMP) game, both players simultaneously choose one of $n$ actions. The payoff to player 1 is $n-1$ if the actions match, or $-1$ if they are different. In a $k$-GMP game (Figure \ref{fig:example_game}), a chance node first selects an index $j$ between $1$ and $k$, and then the players play the $j$'th of $k$ identical GMP games. The following proposition provides a tighter performance bound for XDO in this case, $2n$ iterations instead of $\sum_i |\mathcal{I}_i| = 2k(n + 1)$ (there are $kn$ terminal infostates for each player). For PSRO on $k$-GMP, no tighter bound than the $2n^k$ indicated by Proposition \ref{prop:psro} is known.

\begin{proposition}\label{prop:kgmp}
In $k$-GMP with $n$ actions, XDO terminates in $2n$ iterations.
\end{proposition}


\paragraph{Size of the restricted game.} The number of iterations in each algorithm does not provide the full picture of their performance, since iterations can require vastly different computation times. Intuitively, the restricted game in XDO is much larger than in PSRO when both algorithms have the same population size, because XDO induces an extensive-form restricted game with all discovered actions, while PSRO induces a normal-form restricted game with population policies as actions. However, as both algorithms progress, the XDO restricted game is bounded in size by the original game, while PSRO can induce a game with exponentially many actions.

\paragraph{XDO for sparse-support policies.} XDO is useful when the policies in the population do not cover the full original game, because when they do, finding the restricted game meta-NE is as hard as solving the original game. The motivation behind XDO is that, in games where the NE policies are supported by few actions in most infostates, XDO has the potential to quickly find these actions and terminate without expanding the full game.

To analyze this behavior, consider the $m$-clone GMP game, in which there are $mn$ actions partitioned into $n$ equal classes. The actions of the two players are considered a match (with payoff $n - 1$ to player 1) if they belong to the same class. In $(k, m)$-clone GMP, a chance node selects among $k$ identical $m$-clone GMP games. The following proposition states that in $(k, m)$-clone GMP with $n$ classes, XDO terminates after adding at most $2n$ actions for each player, instead of the full game of $kmn$ actions.


\begin{proposition}
In $(k, m)$-clone GMP with $n$ classes, XDO adds at most $2n$ actions for each player.
\end{proposition}


\paragraph{PSRO lower bound.} Similarly to XDO, PSRO can also outperform the guarantee of Proposition \ref{prop:psro} in certain cases. Generically, however, the linear upper bound on XDO established by Proposition \ref{prop:xdo}, $\sum_i |\mathcal{I}_i|$, is also a \emph{lower bound} on the normal-form population size of pure strategies that is needed to support a NE in PSRO. To show this, consider a perturbed $k$-GMP game, in which the payoffs in each GMP game are slightly modified to induce $k$ distinct NE. The following proposition establishes a linear lower bound for PSRO in perturbed $k$-GMP games.

\begin{proposition}
There exist perturbed $k$-GMP games with $n$ actions in which PSRO cannot terminate in fewer than $k(n - 1) + 1$ iterations.
\end{proposition}

\section{Neural Extensive-Form Double Oracle (NXDO)}

Neural Extensive-Form Double Oracle (NXDO) extends XDO to large games through deep reinforcement learning (DRL). Instead of using an oracle best response, NXDO instead uses approximate best responses that are trained via any DRL algorithm, such as PPO~\citep{schulman2017proximal} or DDQN~\citep{van2016deep}. 
Instead of representing the restricted game explicitly as the set of allowed actions in every infostate, to create its restricted game, NXDO replaces the original game action space with a discrete set of meta-actions, each corresponding to a population policy to which the actual action choice is delegated. 

Formally, NXDO (Algorithm \ref{alg:NXDO}) keeps a population of DRL policies $\Pi^t$ at time $t$. Each iteration, a restricted extensive-form game is created and a meta-NE to the restricted game is computed. The restricted game has meta-actions at every infostate that pick one policy from the population 
\begin{equation}\label{nxdo_restricted_game}
    \forall s_i \in \mathcal{I}_i \quad \mathcal{A}^r_i(s_i) = \{1, 2, ... , |\Pi_i^t|\}.
\end{equation}

\begin{algorithm}[t]
    \caption{NXDO}
    \label{alg:NXDO} 
\begin{algorithmic}[1]

\algblockdefx{MRepeat}{EndRepeat}{\textbf{repeat}}{}
\algnotext{EndRepeat}

\State Input: initial population $\Pi^0$
\MRepeat
    \State Define restricted game for $\Pi^t$ via eq. \eqref{nxdo_restricted_game}
    \State Get $\epsilon$-NE policy $\pi^{r*}$ of restricted game via NFSP
    \State Find $\mathbb{BR}_i(\pi^{r*}_{-i})$ for $i \in \{1,2\}$ via DRL
    \State $\Pi^{t+1}_i = \Pi^t_i \cup \mathbb{BR}_i(\pi^{r*}_{-i})$ for $i \in \{1, 2\}$
\EndRepeat
\end{algorithmic}
\end{algorithm}

While the action space differs, the restricted game states, observations, and histories remain the same as in the original game. 
After each player selects a meta-action that indicates a population policy, an action is sampled from that population policy and used for the world state transition.
With $\pi_i^1, \ldots, \pi_i^{|\Pi_i|}$ the population policies for player $i$, the transition function in the restricted game satisfies
\begin{equation}
    \mathcal{T}^r(h, a^r, w') = \sum_a \prod_i \pi_i^{a_i^r}(s_i(h), a_i) \mathcal{T}(h, a, w').
\end{equation}


With the restricted game thus defined, an $\epsilon$-meta-NE $\pi^{r*}$ is computed in this restricted game via a DRL method for finding NE, such as NFSP~\citep{nfsp} or DREAM~\citep{steinberger2020dream}. 
Approximate BRs $\mathbb{BR}_1(\pi^{r*}_2)$ and $\mathbb{BR}_2(\pi^{r*}_1)$ to this meta-NE are computed via a DRL algorithm, such as PPO or DDQN. These BRs are then added to the population of policies: $\Pi^{t+1}_i = \Pi^t_i \cup \mathbb{BR}_i(\pi^{r*}_{-i})$ for $i \in \{1, 2\}$. Provided that the DRL best responses are sufficiently close to oracle best responses and the inner-loop solver finds a sufficiently close approximate NE of the restricted game, NXDO inherits the same convergence properties as XDO. In practice, contemporary DRL methods lack any guarantee of providing approximate NE or BRs. Nevertheless, we show experimentally that exploitability can decrease through execution of NXDO faster than it does for PSRO and NFSP.

Because the original game action space has no influence on the NXDO restricted game action space, NXDO, like PSRO, is compatible with extremely large and continuous action spaces, provided that the deep RL BRs can operate in such an action space well. We demonstrate this capability in our Loss Game experiments, in which NXDO and PSRO use continuous-action PPO BRs. While no convergence guarantees are known in continuous-action games, NXDO and PSRO empirically produce meta-Nash strategies that are hard to exploit in our experiments.

A drawback of meta-actions that delegate actions to population policies is that, as in PSRO, the number of meta-actions grows linearly with the number of iterations. This can eventually make the restricted game harder to solve than the original game. 
In our experiments, however, NXDO achieves significant improvements in exploitability within a very small number of iterations, such that the issue of action delegation does not become an obstacle. In discrete action space games where it is tractable, we also consider a variant, NXDO-VA, where the restricted game is explicitly calculated and defined with valid and invalid original-game actions in the same way as with tabular XDO, using equation \eqref{restricted_game}. Because its restricted game action space size is at most equal to that of the original game, NXDO-VA does not suffer from the aforementioned drawback.

\section{Experiments}

\begin{figure}
    \centering
    \begin{subfigure}{0.3\textwidth}
        \centering
        \includegraphics[width=1\textwidth]{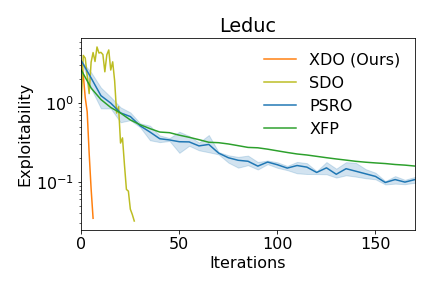}
        \caption{}
        \label{fig:xdo_psro}
    \end{subfigure}
    ~
    \begin{subfigure}{0.3\textwidth}
        \centering
        \includegraphics[width=1\textwidth]{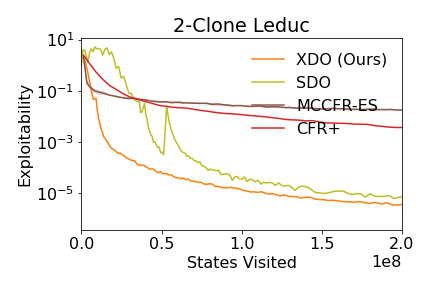}
        \caption{}
        \label{fig:dummy_leduc_info}
    \end{subfigure}
    ~
    \begin{subfigure}{0.3\textwidth}
        \centering
        \includegraphics[width=1\textwidth]{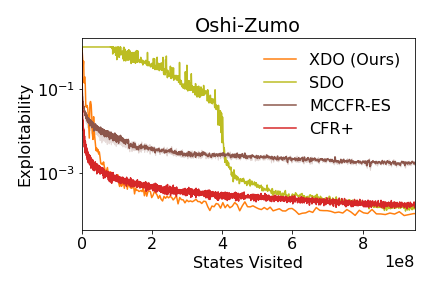}
        \caption{}
        \label{fig:oshi_zumo}
    \end{subfigure}
    
    \label{fig:xdo_cfr_dummy_leduc}
    \caption{(a) Exploitability in Leduc poker of XDO vs. PSRO, SDO, and XFP with oracle BRs throughout their iterations; (b) Exploitability in 2-Clone Leduc poker as a function of the number of game states visited by XDO, SDO, CFR$^+$, and MCCFR with external sampling (ES); (c) Exploitability in Oshi-Zumo as a function of the number of game states visited by XDO, SDO, CFR$^+$, and \mbox{MCCFR-ES}}
\end{figure}

For the tabular experiments, we use XDO with an oracle best response (BR) and CFR$^+$~\citep{DBLP:journals/corr/Tammelin14} for the inner-loop meta-NE solver. We compare XDO with PSRO~\citep{psro} and XFP~\citep{xfp}, which use oracle BRs as well. We also compare with CFR$^+$, and for both CFR$^+$ and XFP we follow the implementations in OpenSpiel \citep{lanctot2019openspiel}. We compare to SDO, using the same oracle BR and meta-NE solver as XDO. Since CFR$^+$, XFP, SDO, and XDO are deterministic, we do not plot error bars for these algorithms. For the neural experiments, we use NXDO with NFSP~\citep{nfsp} as the meta-NE solver and PSRO with FP~\citep{fp} as the meta-NE solver. NXDO and PSRO share the same BR configuration, using DDQN~\citep{van2016deep} for discrete action spaces and PPO~\citep{schulman2017proximal} for continuous action spaces. We compare these algorithms on $m$-Clone Leduc poker, Oshi-Zumo, and the Loss Game, described in the supplementary materials.

\subsection{Tabular Experiments with XDO}


Finding a normal-form meta-NE can be much less efficient and more exploitable than finding an extensive-form meta-NE. This means that PSRO will often require many more pure strategies to achieve a similar level of exploitability to XDO. Figure \ref{fig:xdo_psro} summarizes the results of running XDO, SDO, XFP, and PSRO with an oracle BR on Leduc poker. Even after 150 iterations, PSRO remains significantly more exploitable than XDO is at 7 iterations. XDO achieves exploitability of 0.1 in over 20x fewer iterations than PSRO. In large games where calculating many approximate BRs via reinforcement learning is expensive, requiring vastly more iterations can render PSRO infeasible. XFP performs similarly to PSRO in Leduc, as its average policy is equivalent to uniformly mixing population strategies at the root of the game. SDO takes more iterations than XDO to achieve low exploitability, because it adds fewer actions to the restricted game per iteration.



We compare XDO with SDO, CFR$^+$, and MCCFR \citep{lanctot2009monte} in 2-Clone Leduc poker. In Figure \ref{fig:dummy_leduc_info}, we plot the exploitability of these algorithms as a function of the number of game states visited by the algorithms. Since XFP and PSRO only use BR oracles, we do not include them in this analysis. Since CFR$^+$ updates every infostate every iteration, as we increase the number of cloned actions, the performance of CFR$^+$ will deteriorate. In contrast, XDO will tend to not add cloned actions, which allows the inner-loop CFR$^+$ to expand fewer infostates. These results for XDO are with a deterministic best response oracle. We found that if XDO randomly chose a best response instead, then XDO would still outperform CFR$^+$, but not by as much. The results in Figure \ref{fig:oshi_zumo} suggest that on Oshi-Zumo XDO outperforms SDO, MCCFR, and, to a lesser degree, CFR$^+$. We do not include XDO with MCCFR as the inner loop solver because it did not perform as well as XDO with CFR$^+$. 



\subsection{Neural Experiments with NXDO}

In Figure \ref{fig:20clone_leduc_nxdo}, we compare the exploitability of NXDO and NXDO-VA with PSRO and NFSP on 20-Clone Leduc poker. DDQN is used to train NXDO and PSRO BRs. Similar to the tabular experiments, we find that NXDO outperforms both methods. However, we find that the margin by which NXDO outperforms these methods is smaller than in tabular experiments. This can be attributed to the large proportion of experience required by NFSP to solve the restricted game relative to experience spent learning BRs. Training details and an analysis on the proportion of experience spent in the inner vs outer loop of NXDO are included in supplementary materials.

We also test NXDO and PSRO on the 2D and 16D continuous-action Loss Game, shown in Figures \ref{fig:2d_loss_game} and \ref{fig:16d_loss_game} respectively. PPO is used to train NXDO and PSRO BRs. In the 2D action space game, we compare to NFSP with a binned discrete action space. Because calculating exact exploitability is intractable for the Loss Game, for each algorithm checkpoint, we measure approximate exploitability by training a continuous-action PPO best response against it until saturation and measuring the BR's final mean reward. NXDO outperforms NFSP in the 2D game while operating in a continuous action space and outperforms PSRO in the 2D and 16D game which we conjecture is due to more effective use of population strategies in its restricted game.

\begin{figure}
    \centering
    \begin{subfigure}{0.3\textwidth}
        \centering
        \includegraphics[width=1\textwidth]{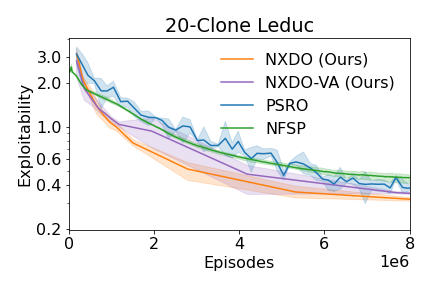}
        \caption{}
        \label{fig:20clone_leduc_nxdo}
    \end{subfigure}
    ~
    \begin{subfigure}{0.3\textwidth}
        \centering
        \includegraphics[width=1\textwidth]{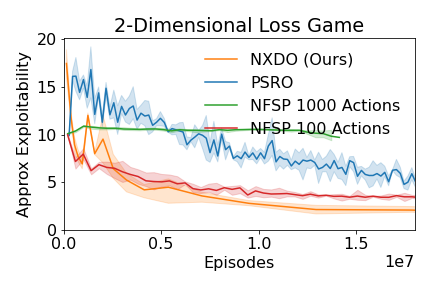}
        \caption{}
        \label{fig:2d_loss_game}
    \end{subfigure}
    ~
    \begin{subfigure}{0.3\textwidth}
        \centering
        \includegraphics[width=1\textwidth]{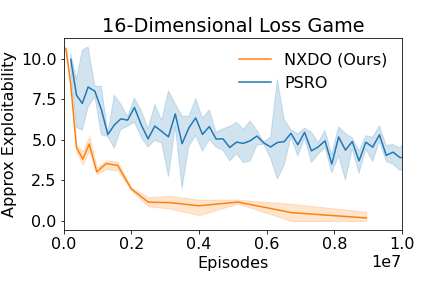}
        \caption{}
        \label{fig:16d_loss_game}
    \end{subfigure}
    \label{fig:searchnodes}
    \caption{(a) Exploitability in 20-Clone Leduc poker of NXDO, NXDO-VA, PSRO, and NFSP as a function of episodes gathered; (b and c) Approximate exploitability as a function of episodes gathered in the continuous-action Loss Game of NXDO, PSRO, and NFSP (with binned discrete actions).}
\end{figure}





\section{Conclusion}


PSRO and NXDO are the only existing game-theoretic deep RL methods that can work on large continuous-action games. In games where the NE must mix in many infostates, but only a small fraction of all actions are in the support of the NE at each infostate, we expect XDO and NXDO to outperform PSRO, because PSRO may require a superlinear, or even exponential number of pure strategies.
We also expect XDO and NXDO to outperform CFR and NFSP, respectively, on discrete-action games where the NE only needs to mix over a small fraction of actions. This is because CFR and NFSP scale poorly with the number of actions in the game, but XDO and NXDO tend to discover a set of relevant actions and ignore actions that are dominated or redundant. We hypothesize that games with these properties are prevalent across a number of domains such as large board games, video games, and robotics applications. 

\section{Acknowledgements}

The authors would like to thank Marc Lanctot and Julien Perolat for interesting and helpful discussions, and Chinmay Tyagi for help in prototyping related concepts.

Roy Fox's research is partly funded by the Hasso Plattner Foundation. 


\bibliographystyle{icml2021}
\bibliography{bibliography}

\begin{thebibliography}{37}
\providecommand{\natexlab}[1]{#1}
\providecommand{\url}[1]{\texttt{#1}}
\expandafter\ifx\csname urlstyle\endcsname\relax
  \providecommand{\doi}[1]{doi: #1}\else
  \providecommand{\doi}{doi: \begingroup \urlstyle{rm}\Url}\fi

\bibitem[Bansal et~al.(2017)Bansal, Pachocki, Sidor, Sutskever, and
  Mordatch]{bansal2017emergent}
Bansal, T., Pachocki, J., Sidor, S., Sutskever, I., and Mordatch, I.
\newblock Emergent complexity via multi-agent competition.
\newblock \emph{arXiv preprint arXiv:1710.03748}, 2017.

\bibitem[Berner et~al.(2019)Berner, Brockman, Chan, Cheung, Debiak, Dennison,
  Farhi, Fischer, Hashme, Hesse, et~al.]{dota}
Berner, C., Brockman, G., Chan, B., Cheung, V., Debiak, P., Dennison, C.,
  Farhi, D., Fischer, Q., Hashme, S., Hesse, C., et~al.
\newblock Dota 2 with large scale deep reinforcement learning.
\newblock \emph{arXiv preprint arXiv:1912.06680}, 2019.

\bibitem[Bosansky et~al.(2014)Bosansky, Kiekintveld, Lisy, and
  Pechoucek]{bosansky2014exact}
Bosansky, B., Kiekintveld, C., Lisy, V., and Pechoucek, M.
\newblock An exact double-oracle algorithm for zero-sum extensive-form games
  with imperfect information.
\newblock \emph{Journal of Artificial Intelligence Research}, 51:\penalty0
  829--866, 2014.

\bibitem[Brown(1951)]{fp}
Brown, G.~W.
\newblock Iterative solution of games by fictitious play.
\newblock \emph{Activity analysis of production and allocation}, pp.\
  374--376, 1951.

\bibitem[Brown \& Sandholm(2015{\natexlab{a}})Brown and
  Sandholm]{brown2015regret}
Brown, N. and Sandholm, T.
\newblock Regret-based pruning in extensive-form games.
\newblock In \emph{NIPS}, pp.\  1972--1980, 2015{\natexlab{a}}.

\bibitem[Brown \& Sandholm(2015{\natexlab{b}})Brown and
  Sandholm]{brown2015simultaneous}
Brown, N. and Sandholm, T.
\newblock Simultaneous abstraction and equilibrium finding in games.
\newblock In \emph{Twenty-fourth international joint conference on artificial
  intelligence}, 2015{\natexlab{b}}.

\bibitem[Brown et~al.(2017)Brown, Kroer, and Sandholm]{brown2017dynamic}
Brown, N., Kroer, C., and Sandholm, T.
\newblock Dynamic thresholding and pruning for regret minimization.
\newblock In \emph{Proceedings of the AAAI Conference on Artificial
  Intelligence}, volume~31, 2017.

\bibitem[Brown et~al.(2019)Brown, Lerer, Gross, and Sandholm]{deep_cfr}
Brown, N., Lerer, A., Gross, S., and Sandholm, T.
\newblock Deep counterfactual regret minimization.
\newblock In \emph{International Conference on Machine Learning}, pp.\
  793--802, 2019.

\bibitem[{\v{C}}erm{\'a}k et~al.(2017){\v{C}}erm{\'a}k, Bo{\v{s}}ansky, and
  Lisy]{vcermak2017algorithm}
{\v{C}}erm{\'a}k, J., Bo{\v{s}}ansky, B., and Lisy, V.
\newblock An algorithm for constructing and solving imperfect recall
  abstractions of large extensive-form games.
\newblock In \emph{Proceedings of the 26th International Joint Conference on
  Artificial Intelligence}, pp.\  936--942, 2017.

\bibitem[Foerster et~al.(2018)Foerster, Farquhar, Afouras, Nardelli, and
  Whiteson]{foerster2018counterfactual}
Foerster, J., Farquhar, G., Afouras, T., Nardelli, N., and Whiteson, S.
\newblock Counterfactual multi-agent policy gradients.
\newblock In \emph{Proceedings of the AAAI Conference on Artificial
  Intelligence}, volume~32, 2018.

\bibitem[Gruslys et~al.(2020)Gruslys, Lanctot, Munos, Timbers, Schmid, Perolat,
  Morrill, Zambaldi, Lespiau, Schultz, et~al.]{gruslys2020advantage}
Gruslys, A., Lanctot, M., Munos, R., Timbers, F., Schmid, M., Perolat, J.,
  Morrill, D., Zambaldi, V., Lespiau, J.-B., Schultz, J., et~al.
\newblock The advantage regret-matching actor-critic.
\newblock \emph{arXiv preprint arXiv:2008.12234}, 2020.

\bibitem[Hansen et~al.(2004)Hansen, Bernstein, and
  Zilberstein]{hansen2004dynamic}
Hansen, E.~A., Bernstein, D.~S., and Zilberstein, S.
\newblock Dynamic programming for partially observable stochastic games.
\newblock In \emph{AAAI}, volume~4, pp.\  709--715, 2004.

\bibitem[Hart \& Mas-Colell(2000)Hart and Mas-Colell]{rm}
Hart, S. and Mas-Colell, A.
\newblock A simple adaptive procedure leading to correlated equilibrium.
\newblock \emph{Econometrica}, 68\penalty0 (5):\penalty0 1127--1150, 2000.

\bibitem[Heinrich \& Silver(2016)Heinrich and Silver]{nfsp}
Heinrich, J. and Silver, D.
\newblock Deep reinforcement learning from self-play in imperfect-information
  games.
\newblock \emph{arXiv preprint arXiv:1603.01121}, 2016.

\bibitem[Heinrich et~al.(2015)Heinrich, Lanctot, and Silver]{xfp}
Heinrich, J., Lanctot, M., and Silver, D.
\newblock Fictitious self-play in extensive-form games.
\newblock In \emph{International Conference on Machine Learning}, pp.\
  805--813, 2015.

\bibitem[Jaderberg et~al.(2019)Jaderberg, Czarnecki, Dunning, Marris, Lever,
  Castaneda, Beattie, Rabinowitz, Morcos, Ruderman, et~al.]{pbt}
Jaderberg, M., Czarnecki, W.~M., Dunning, I., Marris, L., Lever, G., Castaneda,
  A.~G., Beattie, C., Rabinowitz, N.~C., Morcos, A.~S., Ruderman, A., et~al.
\newblock Human-level performance in 3d multiplayer games with population-based
  reinforcement learning.
\newblock \emph{Science}, 364\penalty0 (6443):\penalty0 859--865, 2019.

\bibitem[Kingma \& Ba(2015)Kingma and Ba]{DBLP:journals/corr/KingmaB14}
Kingma, D.~P. and Ba, J.
\newblock Adam: {A} method for stochastic optimization.
\newblock In Bengio, Y. and LeCun, Y. (eds.), \emph{3rd International
  Conference on Learning Representations, {ICLR} 2015, San Diego, CA, USA, May
  7-9, 2015, Conference Track Proceedings}, 2015.
\newblock URL \url{http://arxiv.org/abs/1412.6980}.

\bibitem[Kuhn \& Tucker(1953)Kuhn and Tucker]{kuhn1953contributions}
Kuhn, H.~W. and Tucker, A.~W.
\newblock \emph{Contributions to the Theory of Games}, volume~2.
\newblock Princeton University Press, 1953.

\bibitem[Lanctot et~al.(2009)Lanctot, Waugh, Zinkevich, and
  Bowling]{lanctot2009monte}
Lanctot, M., Waugh, K., Zinkevich, M., and Bowling, M.~H.
\newblock Monte carlo sampling for regret minimization in extensive games.
\newblock In \emph{NIPS}, pp.\  1078--1086, 2009.

\bibitem[Lanctot et~al.(2017)Lanctot, Zambaldi, Gruslys, Lazaridou, Tuyls,
  P{\'e}rolat, Silver, and Graepel]{psro}
Lanctot, M., Zambaldi, V., Gruslys, A., Lazaridou, A., Tuyls, K., P{\'e}rolat,
  J., Silver, D., and Graepel, T.
\newblock A unified game-theoretic approach to multiagent reinforcement
  learning.
\newblock In \emph{Advances in Neural Information Processing Systems}, pp.\
  4190--4203, 2017.

\bibitem[Lanctot et~al.(2019{\natexlab{a}})Lanctot, Lockhart, Lespiau,
  Zambaldi, Upadhyay, P\'{e}rolat, Srinivasan, Timbers, Tuyls, Omidshafiei,
  Hennes, Morrill, Muller, Ewalds, Faulkner, Kram\'{a}r, Vylder, Saeta,
  Bradbury, Ding, Borgeaud, Lai, Schrittwieser, Anthony, Hughes, Danihelka, and
  Ryan-Davis]{LanctotEtAl2019OpenSpiel}
Lanctot, M., Lockhart, E., Lespiau, J.-B., Zambaldi, V., Upadhyay, S.,
  P\'{e}rolat, J., Srinivasan, S., Timbers, F., Tuyls, K., Omidshafiei, S.,
  Hennes, D., Morrill, D., Muller, P., Ewalds, T., Faulkner, R., Kram\'{a}r,
  J., Vylder, B.~D., Saeta, B., Bradbury, J., Ding, D., Borgeaud, S., Lai, M.,
  Schrittwieser, J., Anthony, T., Hughes, E., Danihelka, I., and Ryan-Davis, J.
\newblock {OpenSpiel}: A framework for reinforcement learning in games.
\newblock \emph{CoRR}, abs/1908.09453, 2019{\natexlab{a}}.
\newblock URL \url{http://arxiv.org/abs/1908.09453}.

\bibitem[Lanctot et~al.(2019{\natexlab{b}})Lanctot, Lockhart, Lespiau,
  Zambaldi, Upadhyay, P{\'e}rolat, Srinivasan, Timbers, Tuyls, Omidshafiei,
  et~al.]{lanctot2019openspiel}
Lanctot, M., Lockhart, E., Lespiau, J.-B., Zambaldi, V., Upadhyay, S.,
  P{\'e}rolat, J., Srinivasan, S., Timbers, F., Tuyls, K., Omidshafiei, S.,
  et~al.
\newblock Openspiel: A framework for reinforcement learning in games.
\newblock \emph{arXiv preprint arXiv:1908.09453}, 2019{\natexlab{b}}.

\bibitem[Li et~al.(2018)Li, Hu, Ge, Jiang, Qi, and Song]{li2018double}
Li, H., Hu, K., Ge, Z., Jiang, T., Qi, Y., and Song, L.
\newblock Double neural counterfactual regret minimization.
\newblock \emph{arXiv preprint arXiv:1812.10607}, 2018.

\bibitem[Liang et~al.(2018)Liang, Liaw, Nishihara, Moritz, Fox, Goldberg,
  Gonzalez, Jordan, and Stoica]{pmlr-v80-liang18b}
Liang, E., Liaw, R., Nishihara, R., Moritz, P., Fox, R., Goldberg, K.,
  Gonzalez, J., Jordan, M., and Stoica, I.
\newblock {RL}lib: Abstractions for distributed reinforcement learning.
\newblock In Dy, J. and Krause, A. (eds.), \emph{Proceedings of the 35th
  International Conference on Machine Learning}, volume~80 of \emph{Proceedings
  of Machine Learning Research}, pp.\  3053--3062, Stockholmsmässan, Stockholm
  Sweden, 10--15 Jul 2018. PMLR.
\newblock URL \url{http://proceedings.mlr.press/v80/liang18b.html}.

\bibitem[Lowe et~al.(2017)Lowe, Wu, Tamar, Harb, Abbeel, and
  Mordatch]{lowe2017multi}
Lowe, R., Wu, Y.~I., Tamar, A., Harb, J., Abbeel, O.~P., and Mordatch, I.
\newblock Multi-agent actor-critic for mixed cooperative-competitive
  environments.
\newblock In \emph{Advances in neural information processing systems}, pp.\
  6379--6390, 2017.

\bibitem[Majumdar et~al.(2020)Majumdar, Khadka, Miret, Mcaleer, and
  Tumer]{majumdar2020evolutionary}
Majumdar, S., Khadka, S., Miret, S., Mcaleer, S., and Tumer, K.
\newblock Evolutionary reinforcement learning for sample-efficient multiagent
  coordination.
\newblock In \emph{International Conference on Machine Learning}, pp.\
  6651--6660. PMLR, 2020.

\bibitem[McAleer et~al.(2020)McAleer, Lanier, Fox, and
  Baldi]{mcaleer2020pipeline}
McAleer, S., Lanier, J., Fox, R., and Baldi, P.
\newblock Pipeline psro: A scalable approach for finding approximate nash
  equilibria in large games.
\newblock 2020.

\bibitem[McMahan et~al.(2003)McMahan, Gordon, and Blum]{double_oracle}
McMahan, H.~B., Gordon, G.~J., and Blum, A.
\newblock Planning in the presence of cost functions controlled by an
  adversary.
\newblock In \emph{Proceedings of the 20th International Conference on Machine
  Learning (ICML-03)}, pp.\  536--543, 2003.

\bibitem[Rashid et~al.(2018)Rashid, Samvelyan, De~Witt, Farquhar, Foerster, and
  Whiteson]{rashid2018qmix}
Rashid, T., Samvelyan, M., De~Witt, C.~S., Farquhar, G., Foerster, J., and
  Whiteson, S.
\newblock Qmix: Monotonic value function factorisation for deep multi-agent
  reinforcement learning.
\newblock \emph{arXiv preprint arXiv:1803.11485}, 2018.

\bibitem[Schmid et~al.(2014)Schmid, Moravcik, and Hladik]{schmid2014bounding}
Schmid, M., Moravcik, M., and Hladik, M.
\newblock Bounding the support size in extensive form games with imperfect
  information.
\newblock In \emph{Proceedings of the AAAI Conference on Artificial
  Intelligence}, volume~28, 2014.

\bibitem[Schulman et~al.(2017)Schulman, Wolski, Dhariwal, Radford, and
  Klimov]{schulman2017proximal}
Schulman, J., Wolski, F., Dhariwal, P., Radford, A., and Klimov, O.
\newblock Proximal policy optimization algorithms.
\newblock \emph{arXiv preprint arXiv:1707.06347}, 2017.

\bibitem[Steinberger(2019)]{steinberger2019single}
Steinberger, E.
\newblock Single deep counterfactual regret minimization.
\newblock \emph{arXiv preprint arXiv:1901.07621}, 2019.

\bibitem[Steinberger et~al.(2020)Steinberger, Lerer, and
  Brown]{steinberger2020dream}
Steinberger, E., Lerer, A., and Brown, N.
\newblock Dream: Deep regret minimization with advantage baselines and
  model-free learning.
\newblock \emph{arXiv preprint arXiv:2006.10410}, 2020.

\bibitem[Tammelin(2014)]{DBLP:journals/corr/Tammelin14}
Tammelin, O.
\newblock Solving large imperfect information games using {CFR+}.
\newblock \emph{CoRR}, abs/1407.5042, 2014.
\newblock URL \url{http://arxiv.org/abs/1407.5042}.

\bibitem[Van~Hasselt et~al.(2016)Van~Hasselt, Guez, and Silver]{van2016deep}
Van~Hasselt, H., Guez, A., and Silver, D.
\newblock Deep reinforcement learning with double q-learning.
\newblock In \emph{Proceedings of the AAAI Conference on Artificial
  Intelligence}, volume~30, 2016.

\bibitem[Vinyals et~al.(2019)Vinyals, Babuschkin, Czarnecki, Mathieu, Dudzik,
  Chung, Choi, Powell, Ewalds, Georgiev, et~al.]{alphastar}
Vinyals, O., Babuschkin, I., Czarnecki, W.~M., Mathieu, M., Dudzik, A., Chung,
  J., Choi, D.~H., Powell, R., Ewalds, T., Georgiev, P., et~al.
\newblock Grandmaster level in starcraft ii using multi-agent reinforcement
  learning.
\newblock \emph{Nature}, 575\penalty0 (7782):\penalty0 350--354, 2019.

\bibitem[Zinkevich et~al.(2008)Zinkevich, Johanson, Bowling, and Piccione]{cfr}
Zinkevich, M., Johanson, M., Bowling, M., and Piccione, C.
\newblock Regret minimization in games with incomplete information.
\newblock In \emph{Advances in neural information processing systems}, pp.\
  1729--1736, 2008.

\end{thebibliography}

\newpage
\appendix

\section{Proofs}

\begin{proposition}
In XDO with an $\epsilon_1$-BR oracle, let $\pi^{r*}$ be the final $\epsilon_2$-NE in the restricted game. Then $\pi^{r*}$ is an $(\epsilon_1 + \epsilon_2)$-NE in the full game.
\end{proposition}
\begin{proof}
For each $i \in \{1, 2\}$, let $\mathbb{BR}^{\epsilon_1}_i(\pi^{r*}_{-i})$ be player $i$'s $\epsilon_1$-BR to $\pi_{-i}^{r*}$ obtained in the last iteration.
By the termination condition 
\begin{align}
    v_i(\pi^{r*}) &\ge v_i(\mathbb{BR}^{\epsilon_1}_i(\pi^{r*}_{-i}), \pi^{r*}_{-i}) - \epsilon_2 \\
    &\ge \max_{\pi'_i} v_i(\pi'_i, \pi^{r*}_{-i}) - \epsilon_1 - \epsilon_2,
\end{align}
where the last inequality follows from $\mathbb{BR}^{\epsilon_1}_i(\pi^{r*}_{-i})$ being an $\epsilon_1$-best response to $\pi^{r*}_{-i}$.
\end{proof}



\begin{proposition}\label{prop:psro}
Normal-form DO terminates in $\sum_i \prod_{s_i \in \mathcal{I}_i} |\mathcal{A}_i(s_i)|$ iterations.
\end{proposition}
\begin{proof}
In each iteration of DO, at least one player adds a new normal-form pure strategy to the population.
The space of pure strategies for player $i$ has size $\prod_{s_i \in \mathcal{I}_i} |\mathcal{A}_i(s_i)|$, because each normal-form pure strategy specifies an action at each infostate for that player.
\end{proof}

\begin{proposition}\label{prop:xdo}
XDO terminates in $\sum_i |\mathcal{I}_i|$ iterations. 
\end{proposition}
\begin{proof}
Consider an infostate $s'_i = (a_i^0, o_i^1, \ldots, a_i^t, o_i^{t+1})$ for player $i$ as \emph{covered} in the restricted game if any of player $i$'s population policies chooses action $a_i^t$ in infostate $s_i = (a_i^0, o_i^1, \ldots, a_i^{t-1}, o_i^t)$. At each but the final iteration, at least one player $i$ has $v_i(\mathbb{BR}_i(\pi^{r*}_{-i}), \pi^{r*}_{-i}) > v_i(\pi^{r*}) + \epsilon$. Since $\pi^{r*}_i$ is an $\epsilon$-BR to $\pi^{r*}_{-i}$ in the restricted game, the BR $\mathbb{BR}_i(\pi^{r*}_{-i})$ must be choosing at least some action $a_i$ at some non-terminal infostate $s_i$ that was not previously chosen by any population policy. Adding this action to the restricted game covers at least one previously uncovered infostate: all infostates $s'_i = (s_i, a_i, o_i)$, for any observation $o_i$. All infostates will therefore be covered in at most $\sum_i |\mathcal{I}_i|$ iterations, at which point the next iteration must terminate.
\end{proof}




\begin{proposition}\label{prop:kgmp}
In $k$-GMP with $n$ actions, XDO terminates in $2n$ iterations.
\end{proposition}
\begin{proof}
In a given iteration, consider the restricted game for a single GMP game. If player 2 is allowed an action that player 1 is not, such an action will be player 2's NE, and player 1's BR will add that action. If player 2 is not allowed an action unavailable to player 1, player 2's BR will be a new action unavailable to player 1, if one exists. Thus at least one of the players add a new action in every GMP game in parallel, until both players add all actions.
\end{proof}

\begin{proposition}
In $(k, m)$-clone GMP with $n$ classes, XDO adds at most $2n$ actions for each player.
\end{proposition}
\begin{proof}
The proof repeats that of Proposition \ref{prop:kgmp}, but considering classes instead of actions, because it does not matter which member of a class is added. Once at least one member of each class is added to the restricted game, the meta-NE has full-game exploitability 0, and XDO terminates. In iterations where a BR for a player does not add a new class, it may add a new action member of an existing class. In total, $2n$ actions may be added for each player.
\end{proof}


\begin{proposition}
There exist perturbed $k$-GMP games with $n$ actions in which PSRO cannot terminate in fewer than $k(n - 1) + 1$ iterations.
\end{proposition}
\begin{proof}
For each policy $\pi_2 \in (\Delta(n))^k$ for player 2, consider the perturbed $k$-GMP game that gives player 1 payoff $\frac1{\pi_2(j, a)}$ for matching action $a$ in stage game $j$. In the NE for this game, player 2 has policy $\pi_2$. This implies that the set of policies that are NE of any perturbed $k$-GMP game has positive $k(n -1)$-dimensional volume.

Consider the space of stochastic policies that can be spanned by mixing a specific population of at most $k(n - 1)$ pure strategies. The dimension of this space is at most $k(n - 1) - 1$. When we consider the union of all such spaces for the finitely many possible populations of this size, this set has zero $k(n - 1)$-dimensional volume.

It follows that there exists a perturbed $k$-GMP game $G$, and a neighborhood around player 2's policy in the NE of $G$, such that no policy in that neighborhood is spanned by any population of $k(n - 1)$ pure strategies. For sufficiently small $\epsilon$, no $\epsilon$-NE for $G$ can be found until PSRO adds at least $k(n - 1) + 1$ pure strategies to its population.
\end{proof}

\section{Comparison to SDO}
To illustrate the difference between the two algorithms, we provide an example run of XDO (Figure \ref{fig:xdo}) and SDO (Figure \ref{fig:sdo}) on the same game that is presented as an example in section 4.1.3 in the SDO paper \cite{bosansky2014exact}.

Like in that work, we represent actions that are in the restricted game by bold arrows. This example demonstrates how SDO creates smaller restricted games than XDO because it only considers infostates that can be reached by compatible sequences. In tabular games, SDO results in a cautious approach that only considers a small subset of infostates in order to prevent adding suboptimal actions to the restricted game. However, as we describe below, this will cause obstacles when trying to scale SDO to large games.

\begin{figure}
    \centering
    \begin{subfigure}{0.6 \textwidth}
        \includegraphics[width=\textwidth]{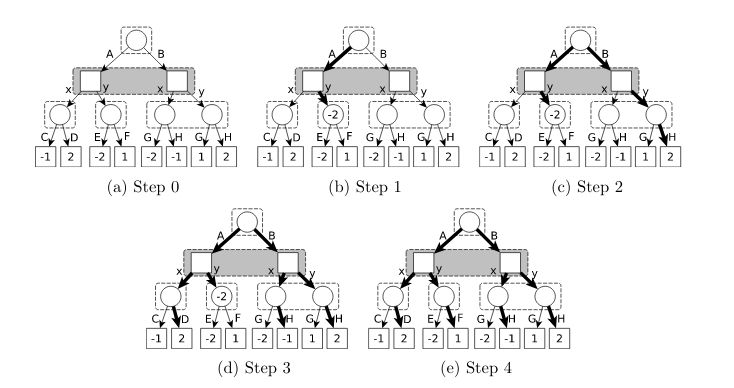}
        \caption{SDO \cite{bosansky2014exact}}
        \label{fig:sdo}
    \end{subfigure}
    \begin{subfigure}{0.39 \textwidth}
        \includegraphics[width=\textwidth]{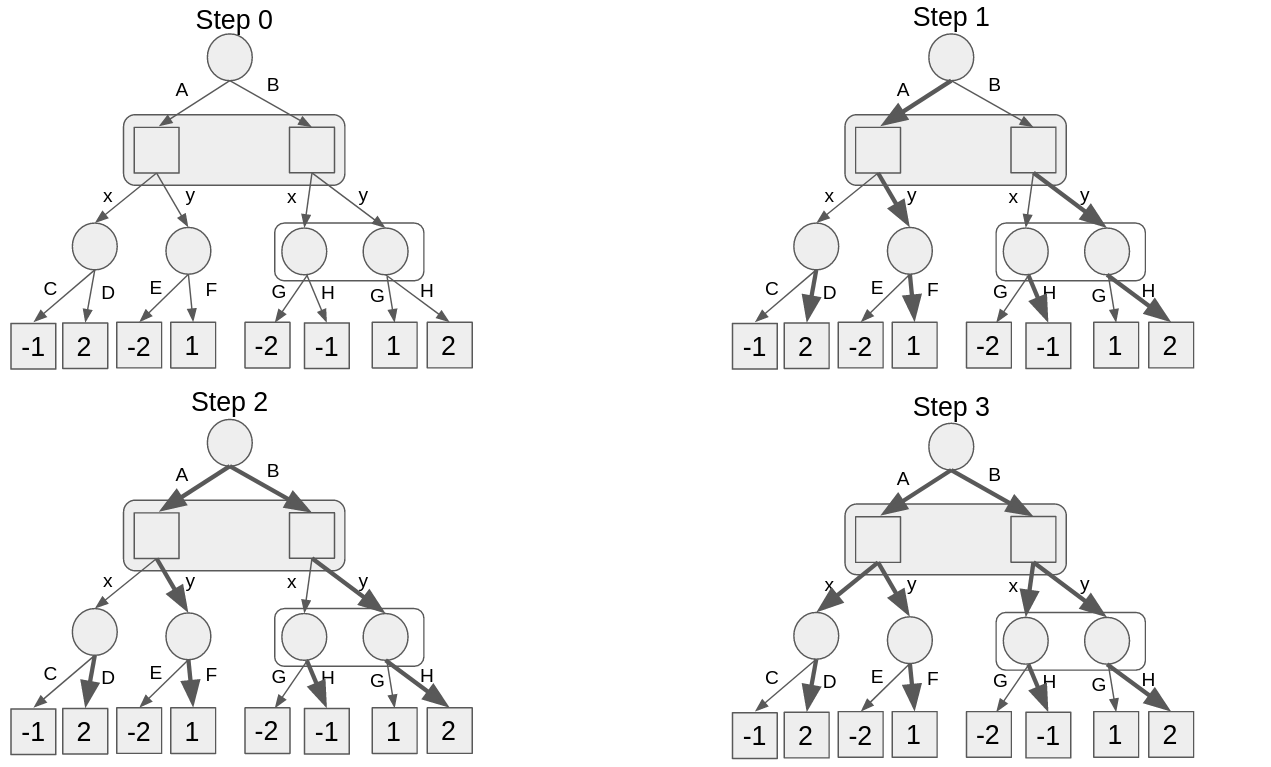}
        \caption{XDO}
        \label{fig:xdo}
    \end{subfigure}
    \label{fig:xdo_sdo}
    \caption{}
\end{figure}

Extensive-form pure strategies specify an action at every infostate. In this example we will refer to extensive-form pure strategies by concatenating the actions the strategy takes in every infostate. For example, the pure strategy for the circle player in step 1 in our diagram corresponds to ADFH. Sequence-form pure strategies specify a sequence of actions that must be internally consistent. For example, $\{\emptyset, A, AD\}$ is a valid sequence-form pure strategy but $\{\emptyset, B, AD\}$ is not.

In step 0, both SDO and XDO start with an empty game tree. Let’s assume that the default strategy for both algorithms is uniform random.

In step 1, both SDO and XDO add the same best responses to the default strategy for both players. However, SDO adds actual sequences of actions, in particular $\{\emptyset, A, AD\}$ for the circle player and $\{\emptyset, y\}$ for the box player. Since the AD sequence of actions for the circle player is not compatible with the y action for the square player, the restricted game is the game with A as the available action for circle and y as the available action for square. But since neither AE nor AF are in the sequence population, the restricted game in SDO at this step terminates in a temporary leaf at that point.

In contrast, XDO adds full extensive form strategies, in this case adding ADFH for the circle player and y for the square player. Now the restricted game for XDO is as in step 1 in the diagram. Since XDO adds full extensive form strategies, there is no need to create a temporary leaf node. 

After only one iteration, SDO and XDO result in a different restricted game. This is because SDO adds sequences of actions (best responses in sequence form) to the population while XDO adds a pure strategy defined at every infostate (best responses in extensive form). The SDO restricted game is much smaller than the XDO restricted game because SDO only considers infostates that can be reached by compatible sequences in the population. There are three more steps for SDO, which are described in their paper. XDO has only two more steps, which are described in the figure. 

There currently exists one main algorithm that can scale up the double oracle approach to large games through deep reinforcement learning, which is PSRO. We propose another, Neural Extensive-Form Double Oracle (NXDO). 
Similarly extending SDO to large games via neural networks is one possible direction for future work.




\section{Empirical Tests on Perturbed $k$-GMP}
\begin{figure}[H]
    \captionsetup[subfigure]{labelformat=empty}

    \centering
    \begin{subfigure}{0.23\textwidth}
        \centering
        \includegraphics[width=1\textwidth]{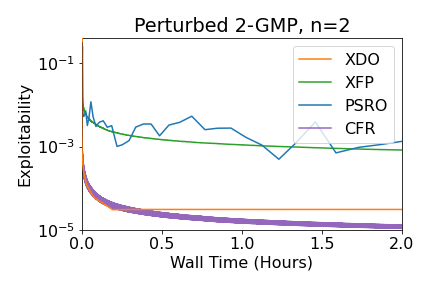}
        \caption{}
        \label{fig:64-gmp}
    \end{subfigure}
    ~
    \begin{subfigure}{0.23\textwidth}
        \centering
        \includegraphics[width=1\textwidth]{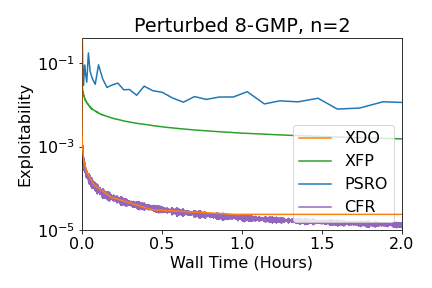}
        \caption{}
        \label{fig:128-gmp}
    \end{subfigure}
    ~
    \begin{subfigure}{0.23\textwidth}
        \centering
        \includegraphics[width=1\textwidth]{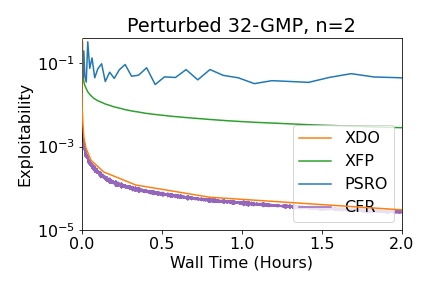}
        \caption{}
        \label{fig:256-gmp}
    \end{subfigure}
    ~
    \begin{subfigure}{0.23\textwidth}
        \centering
        \includegraphics[width=1\textwidth]{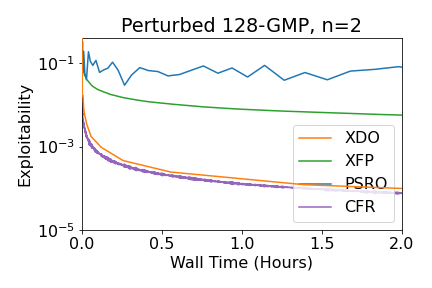}
        \caption{}
        \label{fig:512-gmp}
    \end{subfigure}
    \label{fig:k-gmp}
    \caption{Exploitability vs wall time hours with XDO, XFP, PSRO, and CFR in Perturbed $k$-GMP as the number of subgames rises. XDO and CFR scale well even when there are many subgames while PSRO becomes unable to reach a low exploitability.}
\end{figure}

We run XDO, XFP, PSRO (DO) and CFR on the Perturbed $k$-GMP games. In perturbed $k$-GMP, a chance node randomly selects one of the $k$ stage games and the outcome is observed by both players. A static perturbation uniformly sampled between $(-1.0, 1.0)$ is added to every GMP payoff where both players select the same action to create a unique equilibrium solution for each of the $k$ subgames. We see that as $k$ is increased from 2 to 128, XDO and CFR exploitability increases but remains relatively low, but PSRO performance relative to computation time dramatically decreases. For this experiment, we use CFR as the meta-NE solver for XDO. 

\section{Game Descriptions}
\paragraph{$\bm{m}$-Clone Leduc poker:}$m$-Clone Leduc poker is similar to Leduc poker but with every action duplicated $m$ times, such that instead of a single call, fold, and bet action there are $m$ identical call, fold, and bet actions. As the number of cloned actions increases, we expect the performance of methods based on CFR and FP such as DREAM and NFSP to deteriorate, while the performance of XDO and PSRO remains largely unchanged because the size of the restricted games scale with the size of the meta-game population rather than the size of the action space. 

\paragraph{Oshi-Zumo} Oshi-Zumo is a zero-sum two-player sequential game where both players try to move a token to the other player's side of the board. Both players start with an allotment of coins and each step simultaneously bet an amount of coins. The player who bets the higher amount gets to move the token toward the opponent. A player wins if they are able to move the token off the board on their opponent's side or if the token is on the opponent's side when either no more coins remain or the time horizon is reached. Reward is 1 for wining, -1 for losing, and 0 for a tie, which can occur when the token is in the middle of the board when the game ends. We consider a variant with 4 coins, and a board of 3 spaces with a time horizon of 6. 

\paragraph{Loss Game:} The Loss Game is a zero-sum two-player sequential continuous action game in which agents simultaneously apply bounded adjustments to a real valued vector of parameters $\bm{x} = [x_1,\dots,x_d]$ in order to optimize a fixed loss function's scalar output $L: \mathbb{R}^d \mapsto \mathbb{R}$ . One agent aims to maximize this output while to other aims to minimize it. Each timestep, agents observe the current vector of parameter values along with the function's scalar output value and provide a bounded continuous action describing an adjustment vector to add to the current parameters. The sum of both agents' adjustments is applied and player 1 is rewarded with the value of the function's output while player 2 is provided with the negative of this value. The game lasts 10 steps and the parameters' adjusted values are preserved after each step.

We consider a 2D action space variation with the loss function $L(x_1, x_2) = \sum_j\sin(x_j)$ and a 16D action space variation with $L(x_1,\dots,x_{16}) = \sin(\sum_j x_j) + \sum_j\sin(x_j) / 16$. These functions were chosen to demonstrate games in which it is advantageous to mix between multiple strategies in many infostates. The 2 and 16 dimensional action spaces where chosen to represent continuous spaces that, respectively, can and cannot be directly binned into a tractable amount of discrete actions. We test the binned discrete action version of the 2D continuous action space Loss Game with both 10 and 100 bins in each of the two dimensions , totalling 100 and 10,000 actions respectively.

\section{Tabular Experiment Details}
For calculating tabular BRs, we use an oracle implementation supplied by the OpenSpiel framework \cite{LanctotEtAl2019OpenSpiel}.

For tabular PSRO, we use linear programming to solve each meta-NE. To calculate the empirical payoff matrix, we play 100 games per policy combination. 


For CFR$^+$, CFR, MCCFR, and XFP, we use the default implementations and settings provided by OpenSpiel.


For tabular XDO, we repeatedly improve the $\epsilon$-NE (e.g. perform CFR iterations) for the restricted game until both of the following conditions are met: \textit{a)} the exploitability of the meta-NE in the restricted game is less than $\epsilon$, and \textit{b)} the exploitability of the meta-NE in the restricted game is less than the exploitability of the meta-NE in the full game. We initialize $\epsilon$ to be $0.35$, and each XDO iteration, we set $\epsilon$ to $0.98 * \epsilon$.

This has two benefits: First, we guarantee that all BRs added to the population are useful, in that they contain at least some out-of-restricted-game action, and thus expand the restricted game.  Second, if the restricted game already contains the actions necessary for a NE in the full game, we simply continue to improve the meta-NE ad infinitum, rather than add a new BR and restart the meta-NE solver.

For SDO, we use CFR+ as an inner-loop solver, instead of a linear program as used in the SDO paper \citep{bosansky2014exact}. The default pure strategy we use is to choose the first possible action. To find best-response sequences, we compute a tabular BR in the full game.

Unless otherwise stated, all non-deterministic tabular experiments were run with 3 seeds each.

\section{Neural Experiment Details}

For reinforcement learning neural experiments, we use DDQN for discrete action spaces and PPO for the Loss Game continuous action space. The NFSP RL policy is also trained with DDQN. All neural experiments were run with 3 seeds each.

\subsection{Best Response Stopping Conditions}

In both 20-Clone Leduc Poker and the Loss Game, BRs in PSRO and NXDO are trained for a minimum of 4e4 episodes and a maximum of 1e5 episodes. Every 2e4 episodes, a check is performed as to whether the average episode reward is at least 0.01 higher than it was in the last check. If this check fails, the BR is considered plateaued and is stopped early. This stopping condition was chosen to ensure that BRs are allowed to saturate in performance without spending a redundantly large amount of episodes training.

\subsection{Meta-NE Solver Stopping Conditions}

For NXDO, we use NFSP as the meta-NE solver. Similar to tabular XDO, we increase the amount of experience allocated to solve the meta-NE over multiple NXDO iterations. This is done because adding additional diverse BR policies to the population in early iterations of NXDO is more effective than spending a large (or any) amount of time solving for the restricted game meta-NE with a very small population. In addition to a schedule used to increase the amount of training the for meta-NE solver, for the first $n$ NXDO iterations, we execute a "warm-start" in which we use an untrained NFSP average policy network in lieu of training the meta-NE. After $n$ warm start iterations, we train NFSP to solve for the meta-NE with a schedule that gradually increases the amount of experience NFSP is allocated for training each NXDO iteration. As the NXDO population grows larger, it is advantageous to train the meta-NE solver for longer because the restricted game solution will better approximate the original game NE. We use two different meta-NE stopping condition schedules for 20-Clone Poker and the Loss Game, described below:

The meta-NE stopping condition schedule used for NXDO on 20-Clone Poker solves for an $\epsilon$-NE of the restricted game, where $\epsilon$ is decreased over NXDO iterations. After 7 warm start iterations, $\epsilon$ is initialized at 1.0, and, each NXDO iteration, $\epsilon$ is halved if the final average reward of the previous iteration’s BRs is below $(\epsilon + 0.05)$. $\epsilon$ is not allowed to fall below 0.05. The exact stopping condition for NFSP in the inner loop of NXDO is to train for at least 2e5 episodes and then stop when the mean DDQN BR reward against average policies falls below $\epsilon$.

A simpler schedule is used for NXDO on the Loss Game. After 5 warm start iterations, we train NFSP on the restricted game for 1e6 time steps. In each subsequent NXDO iteration, the amount of time steps allocated to training NFSP meta-NE solver in the new iteration is the previous amount multiplied by a coefficient of 1.5. 

When solving the normal-form restricted game in PSRO, we calculate the meta-NE by running FP for 2000 iterations on the empirical payoff metric between population policies. To calculate the empirical payoff matrix, for each policy combination, we play 3000 games for $m$-Clone Leduc or 1000 games for the Loss Game. Empirical payoff matrix evaluations for PSRO are not counted when measuring experience used to train. 

\subsection{Training Hyperparameters}

For 20-Clone Leduc, we used the OpenSpiel Leduc NFSP implementation as a starting point for our hyperparameters. We then adjusted our learning rate, batch size, and rollout steps to allow better wall-time speed through parallelization while maintaining roughly the same sample efficiency as the OpenSpiel reference. The same hyperparameters are used for the 20-Clone Leduc NXDO meta-NE solver and NFSP on the base 20-Clone Leduc game.

When training for the Loss Game, aside from the network architecture and epsilon greedy annealing period, the same NFSP hyperparameters were reused when training the NXDO meta-NE solver. Because the binned discrete action Loss Game is much more complex than the induced NXDO restricted game, new DDQN hyperparameters were found for NFSP on the original game using a random search of 50 samples over a space defined in table \ref{table:ddqn-lossgame-nfsp-base-search-space}. DDQN hyperparameter trials were performed by training and measuring performance against a fixed adversary.

DDQN uses epsilon-greedy exploration. For 20-Clone Leduc experiments, in NFSP (both on the base game and restricted game), the epsilon exploration parameter is linearly annealed from 0.06 to 0.001 over 20e6 timesteps. In PSRO and NXDO, where BRs train over shorter periods, epsilon is linearly annealed from 0.20 to 0.001 over 1e5 timesteps. For NFSP in the Loss Game (both on the base game and restricted game) the epsilon exploration parameter is annealed from 0.06 to 0.001 over 500e6 timesteps. For all DDQN uses, learning does not start until 16,000 steps have been collected in the circular replay buffer.

Hyperparameters for Loss Game PPO BRs used by NXDO and PSRO were found in a random search of 50 samples over a space defined in table \ref{table:ppo-lossgame-search-space}. Like DDQN, PPO hyperparameter trials were performed by training and measuring performance against a fixed adversary. The same hyperparameters were used in both variations of the Loss Game.

Approximate exploitability in the Loss Game was measured by training a continuous-action PPO BR against each meta-NE  (for NXDO and PSRO) or average policy (for NFSP) checkpoint and measuring the final average reward. These BRs share the same hyperparameters and stopping conditions as those used to train NXDO and PSRO.

Tables \ref{table:ddqn-20clone} through \ref{Tab:nfsp-lossgame-base} display parameters used. Our deep RL implementations for NFSP, PSRO, and NXDO are built on top of the open source RLlib framework. \cite{pmlr-v80-liang18b}. Any parameters not listed were set to the RLlib 1.0.1 default values.

All RL algorithms use the Adam optimizer \cite{DBLP:journals/corr/KingmaB14} with $\epsilon$ set to 1e-8 and 4 parrallel experience workers with 1 environment per worker. A discount factor of 1.0 is used in all games. For 20-Clone Poker, all algorithms used a fully-connected network with two layers of size 128 and relu activations. For the Loss Game, two layers of size 32 and tanh activations were used.


\begin{table}[H]
\centering
\begin{tabular}{ll}
circular buffer size & 2e5 \\
total rollout experience gathered each iter & 1024 steps \\
learning rate & 0.01 \\
batch size & 4096 \\
TD-error loss type & MSE \\
target network update frequency & every 10,000 steps \\
\end{tabular}
\caption{20-clone Leduc DDQN hyperparameters (Used by NXDO BRs and meta-NE solver, PSRO BRs, base game NFSP)}
\label{table:ddqn-20clone}

\end{table}

\begin{table}[H]
\centering
\begin{tabular}{ll}
RL learner params & DDQN, see Table \ref{table:ddqn-20clone} \\
anticipatory param & 0.1  \\
avg policy reservoir buffer size & 2e6 \\
avg policy learning starts after & 16,000 steps \\
avg policy learning rate & 0.1 \\
avg policy batch size & 4096 \\
avg policy optimizer & SGD \\
\end{tabular}
\label{Tab:nfsp-20clone}
\caption{20-clone Leduc NFSP hyperparameters (Used by base game NFSP, NXDO meta-NE solver)}
\end{table}

\begin{table}[H]
\centering
\begin{tabular}{ll}
GAE $\lambda$ & 1.0 \\
entropy coeff & 0.01 \\
clip param & 0.2 \\
KL target & 0.01 \\
learning rate & 5e-4 \\
train batch size & 2048 \\
sgd minibatch size & 256 \\
num sgd iters on train batch & 30 \\
separate policy and value networks & Yes \\
continuous action range & [-1.0, 1.0] for each dim
\end{tabular}
\label{Tab:ppo-params}
\caption{Loss Game Continuous Action PPO hyperparameters (Used by NXDO BRs, PSRO BRs)}
\end{table}

\begin{table}[H]
\centering
\begin{tabular}{ll}
circular buffer size & 2e5 \\
total rollout experience gathered each iter & 1024 steps \\
learning rate & 0.01 \\
batch size & 4096 \\
TD-error loss type & MSE \\
target network update frequency & every 10,000 steps \\
\end{tabular}
\caption{Loss Game NXDO meta-NE DDQN hyperparameters (Used by NXDO meta-NE solver)}
\label{table:ddqn-lossgame-nxdo}

\end{table}

\begin{table}[H]
\centering
\begin{tabular}{ll}
RL learner params & DDQN, see Table \ref{table:ddqn-lossgame-nxdo} \\
anticipatory param & 0.1  \\
avg policy reservoir buffer size & 2e6 \\
avg policy learning starts after & 16,000 steps \\
avg policy learning rate & 0.1 \\
avg policy batch size & 4096 \\
avg policy optimizer & SGD \\
\end{tabular}
\caption{Loss Game NXDO meta-NE NFSP hyperparameters (Used by NXDO meta-NE solver)}
\label{Tab:nfsp-nxdo-lossgame}
\end{table}

\begin{table}[H]
\centering
\begin{tabular}{ll}
circular buffer size & 1e5 \\
learning starts after & 16,000 steps \\
total rollout experience gathered each iter & 64 steps \\
learning rate & 0.007 \\
batch size & 4096 \\
TD-error loss type & MSE \\
target network update frequency & every 1e5 steps \\
\end{tabular}
\caption{Loss Game NFSP DDQN hyperparameters (Used by base game NFSP)}
\label{table:ddqn-lossgame-nfsp-base}

\end{table}

\begin{table}[H]
\centering
\begin{tabular}{ll}
RL learner params & DDQN, see Table \ref{table:ddqn-lossgame-nfsp-base} \\
anticipatory param & 0.1  \\
avg policy reservoir buffer size & 2e6 \\
avg policy learning starts after & 16,000 steps \\
avg policy learning rate & 0.07 \\
avg policy batch size & 4096 \\
avg policy optimizer & SGD \\
\end{tabular}
\caption{Loss Game NFSP hyperparameters (Used by base game NFSP)}
\label{Tab:nfsp-lossgame-base}
\end{table}

\begin{table}[H]
\centering
\begin{tabular}{ll}
circular buffer size & \{5e4, 1e5, 2e5\} \\
total rollout steps gathered each iter & \{16, 32, 64, 128\}\\
learning rate & log-uniform([0.0001, 0.1]) \\
batch size & \{1024, 2048, 4096\} \\
target network update frequency & \{1e3, 1e4, 1e5\} \\
\end{tabular}
\caption{Loss Game DDQN hyperparameter search space (Used by base game NFSP)}
\label{table:ddqn-lossgame-nfsp-base-search-space}

\end{table}

\begin{table}[H]
\centering
\begin{tabular}{ll}
GAE $\lambda$ & \{0.9, 1.0\} \\
entropy coeff & \{0.0, 0.001, 0.01, 0.1\} \\
clip param & \{0.1, 0.2, 0.3\} \\
KL target & \{0.001, 0.01, 0.1\} \\
learning rate & \{5e-2, 5e-3, 5e-4, 5e-5, 5e-6\} \\
train batch size & \{2048, 4096\} \\
sgd minibatch size & \{128, 256, 512, 1024\} \\
num sgd iters on train batch & \{1, 5, 10, 30, 60\} \\
\end{tabular}
\caption{Loss Game Continuous Action PPO hyperparameter search space (Used by NXDO BRs, PSRO BRs)}
\label{table:ppo-lossgame-search-space}
\end{table}

\section{Costs and Benefits of Extensive-Form Restricted Games}

\begin{figure}[H]

    \centering
    \begin{subfigure}{0.3\textwidth}
        \centering
        \includegraphics[width=1\textwidth]{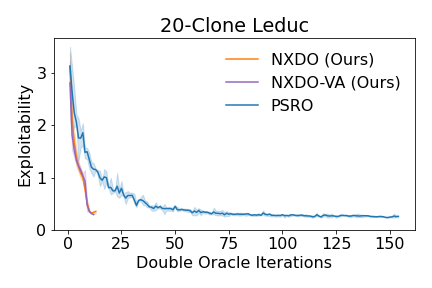}
        \caption{}
        \label{fig:20-leduc-iterations}
    \end{subfigure}
    ~
    \begin{subfigure}{0.3\textwidth}
        \centering
        \includegraphics[width=1\textwidth]{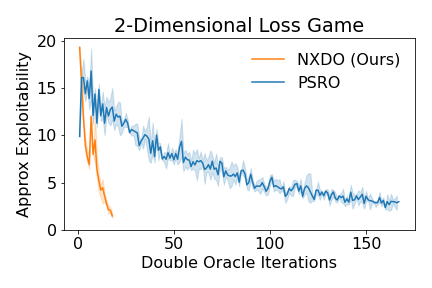}
        \caption{}
        \label{fig:2d-loss-iterations}
    \end{subfigure}
    ~
    \begin{subfigure}{0.3\textwidth}
        \centering
        \includegraphics[width=1\textwidth]{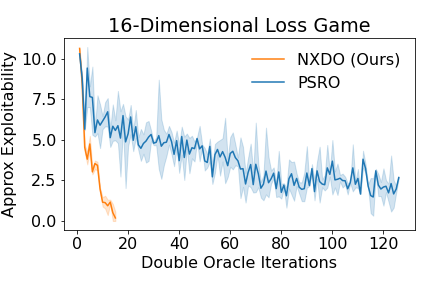}
        \caption{}
        \label{fig:16d-loss-iterations}
    \end{subfigure}
    \caption{NXDO exploitability in 20-Clone Leduc (a) and approximate exploitability in the Loss Game (b and c) as a function of Double Oracle iterations in NXDO and PSRO.}
    \label{fig:do-iterations}

\end{figure}

Shown in figure \ref{fig:do-iterations}, we compare the exploitability of NXDO against that of PSRO in terms of Double Oracle iterations. NXDO acheives a low exploitability in significantly less iterations due to mixing population strategies at every infostate in the game rather than just the at root like is done in PSRO. When both are limited to a small population of behavioral strategies, the extensive-form restricted game allows for much more expressive power in solving for a meta-NE than a normal-form restricted game does.

\begin{figure}[H]

    \centering
    \begin{subfigure}{0.3\textwidth}
        \centering
        \includegraphics[width=1\textwidth]{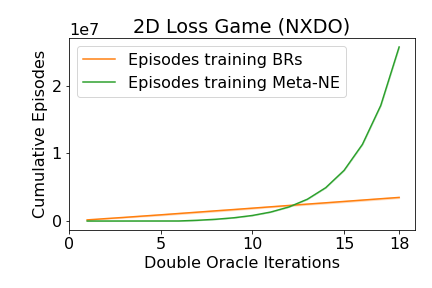}
        \caption{}
        \label{fig:2d-loss-episodes}
    \end{subfigure}
    ~
    \begin{subfigure}{0.3\textwidth}
        \centering
        \includegraphics[width=1\textwidth]{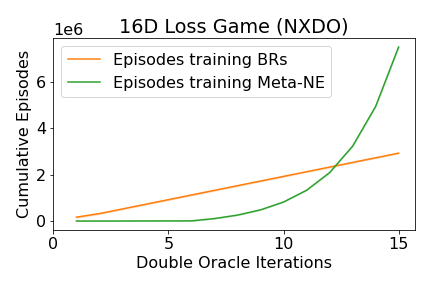}
        \caption{}
        \label{fig:16d-loss-episodes}
    \end{subfigure}
    ~
    \begin{subfigure}{0.3\textwidth}
        \centering
        \includegraphics[width=1\textwidth]{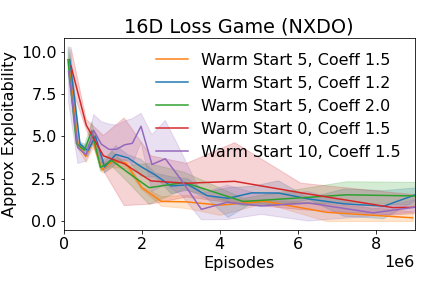}
        \caption{}
        \label{fig:16d-ablations}
    \end{subfigure}

    \label{fig:ablations}
    \caption{(a and b) In NXDO tests on the Loss Game, the cumulative amount of experience in episodes to either train BRs or calculate meta-NE with NFSP as a function of Double Oracle iterations. After the warm start phase, each NXDO iteration, we multiply the amount of episodes we spend training NFSP by a coefficient of 1.5. (c) NXDO ablations on the 16-Dimensional Loss Game. We vary the number of warm start iterations in which we spend zero time training the NFSP meta-NE solver and the coefficient by which we multiply the episodes spent training NFSP each iteration after the warm start.}
\end{figure}

This additional expressive power in the restricted game does come at an increased computational cost which needs to be taken into account. NXDO is most applicable when the need to mix population strategies at many infostates outweighs the extra computation needed to compute the extensive-form restricted game. Figures \ref{fig:2d-loss-episodes} and \ref{fig:16d-loss-episodes} show the cumulative episodes spent training either BRs or meta-NE vs NXDO iterations. The time spent training a meta-NE is gradually increased each iteration to trade quickly adding BRs for solving the meta-NE at a high accuracy.

\section{Neural Methods on Kuhn and Leduc Poker}

\begin{figure}[H]

    \centering
    \begin{subfigure}{0.3\textwidth}
        \centering
        \includegraphics[width=1\textwidth]{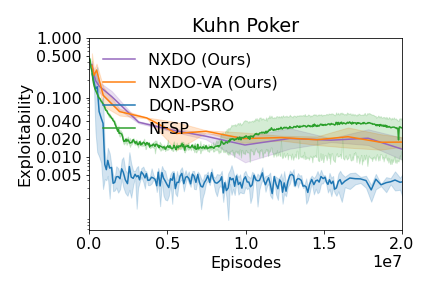}
        \caption{}
        \label{fig:nxdo_kuhn}
    \end{subfigure}
    ~
    \begin{subfigure}{0.3\textwidth}
        \centering
        \includegraphics[width=1\textwidth]{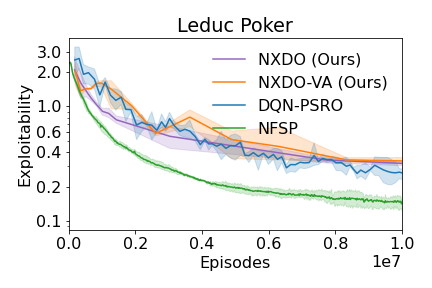}
        \caption{}
        \label{fig:nxdo_leduc}
    \end{subfigure}

    \label{fig:nxdo_kuhn_leduc}
    \caption{NXDO, NXDO-VA, PSRO, and NFSP exploitability vs episodes collected on Kuhn and Leduc Poker}

\end{figure}

We test NXDO, NXDO-VA, PSRO, and NFSP on Kuhn and Leduc Poker using the same hyperparameters as used with 20-Clone Leduc. In these smaller games, NXDO performs less competitively because the time spent computing the extensive-form meta-NE strategies is large compared to the amount of time needed to train nearly all pure strategies for Kuhn and Leduc. NFSP also reaches an initial low exploitability quickly in part due to the small action-space sizes relative to other games tested.

\section{Additional Loss Game Experiments}

We analyze the effect of different stopping condition schedules for the NFSP meta-NE solver when training NXDO on the 16-Dimensional Loss Game. In figure \ref{fig:16d-ablations} we compare a variaty of warm start amounts (the number of initial iterations in which zero time is spent training the meta-NE) and the coefficient at which we increase the number of episodes spent training any meta-NE thereafter. The default parameters are a warm start of 5 and a coefficient of 1.5, thus we spend 5 fixed NXDO iterations with a randomly initialized meta-NE solution, then train our sixth iteration meta-NE for the starting amount of 1e6 episodes, and then train each subsequent meta-NE for 1.5x the number of episodes in the previous iteration. Ablations with other reasonable values considered show that the fine details of such a schedule have only minor effects on performance in the Loss Game.





\begin{table}[h]
\parbox{.5\linewidth}{
\centering
\begin{tabular}{rll}
NXDO against PSRO & $0.26\pm0.06$ & \\
NXDO against NFSP & $1.28\pm0.06$ & \\
PSRO against NFSP & $1.26\pm0.08$ & 
\end{tabular}
\caption{2D Loss Game Round Robin Payoffs}
\label{table:2d-roundrobin}
}
\hfill
\parbox{.5\linewidth}{
\centering
\begin{tabular}{rll}
NXDO against PSRO & $4.43\pm0.06$ &
\end{tabular}
\caption{16D Loss Game Round Robin Payoffs}
\label{table:16d-roundrobin}

}
\end{table}

We also perform a round robin evaluation for the 2-Dimensional and 16-Dimensional Loss Game by measuring empirical payoffs between each algorithm's checkpoints after respectively 2e7 and 1e7 episodes of training. For a given matchup between two algorithms, we play 1000 games, swapping sides every other game, for each pairing of either algorithm's 3 seeds, resulting in 9000 games total. The reward means and 95\% confidence intervals are reported in tables \ref{table:2d-roundrobin} and \ref{table:16d-roundrobin}. NXDO beats PSRO and NFSP in our Loss Game evaluations.

\section{Empirical Analysis of XDO Restricted Game Sizes}
\begin{figure}[H]

    \centering
    \begin{subfigure}{0.3\textwidth}
        \centering
        \includegraphics[width=1\textwidth]{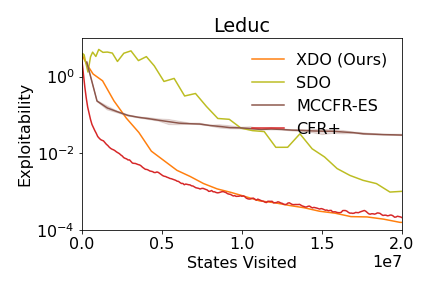}
        \caption{}
        \label{fig:leduc-states-exp}
    \end{subfigure}
    ~
    \begin{subfigure}{0.3\textwidth}
        \centering
        \includegraphics[width=1\textwidth]{neurips_camera_ready_tabular_images/XDO_vs_CFR_vs_MCCFR_Infostates_2_clone_leduc.png}
        \caption{}
        \label{fig:leduc-dummy-states-exp}
    \end{subfigure}

    \label{fig:size-of-restricted-game}
    \caption{XDO exploitability compared to CFR$^+$, SDO and MCCFR as a function of states visited in (a) Leduc poker and (b) 2-Clone Leduc Poker.}
\end{figure}

XDO outperforms CFR$^+$ in games where Nash equilibria require mixing over a small fraction of actions. For example, XDO performs similarly to CFR$^+$ in standard Leduc poker, but XDO significantly outperforms CFR$^+$ in 2-Clone Leduc poker. XDO scales better in this transition because a restricted game sufficient to achieve a low exploitability can be induced with only a portion of the game's actions. Such a restricted game is much smaller in terms of total game states and thus requires less computation to solve. The number of states in 2-Clone Leduc poker is roughly 50 times that of standard Leduc poker, but when moving from standard to 2-Clone Leduc poker, the number of states in the XDO restricted game increases by a much smaller proportion.
    
\begin{figure}[H]

    \centering
    \begin{subfigure}{0.3\textwidth}
        \centering
        \includegraphics[width=1\textwidth]{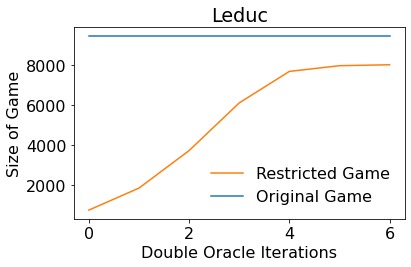}
        \caption{}
        \label{fig:leduc-size}
    \end{subfigure}
    ~
    \begin{subfigure}{0.3\textwidth}
        \centering
        \includegraphics[width=1\textwidth]{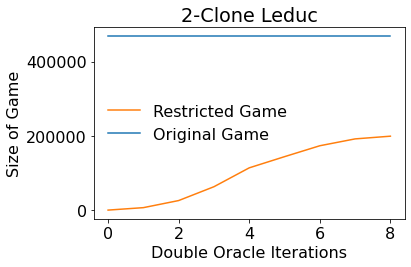}
        \caption{}
        \label{fig:leduc-dummy-size}
    \end{subfigure}
    ~
    \begin{subfigure}{0.3\textwidth}
        \centering
        \includegraphics[width=1\textwidth]{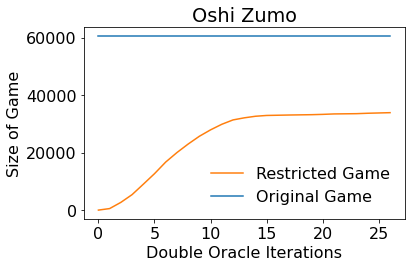}
        \caption{}
        \label{fig:oshi-zumo-size}
    \end{subfigure}

    \label{fig:size-of-restricted-game-2}
    \caption{The size, in states, of the restricted game induced by the BR population in XDO in various games.}
\end{figure}

In our XDO experiments, the final size (in states) of the restricted game in Leduc Poker is roughly 85\% of the size of the full game.  The final size of the restricted game in 2-Clone Leduc Poker is only approximately 43\% that of the full game, so XDO performs proportionally less computation to solve the restricted game meta-NE than it would take to directly solve the full-game NE with all actions considered. XDO also outperforms CFR in Oshi Zumo, where the final size of the restricted game is only half that of the full game.

\section{Computational Costs}
Experiments were performed a shared local computer with 128 CPU-cores, 2 RTX 3090 GPUs, and 256GB of RAM. Due to small network sizes, most neural experiments were performed without GPU acceleration. Neural experiments on 20-Clone Poker and the Loss Game used 8 to 16 cores each and took between 2 and 4 days to complete using 10 to 40GB of RAM. Tabular XDO experiments were run for up to 1 day, using a single core each and between 1 to 10GB of RAM.

\section{Experiment Code}
Code for tabular experiments can be found at \url{https://github.com/indylab/tabular_xdo}

Code for neural experiments can be found at \url{https://github.com/indylab/nxdo}

\end{document}